\documentclass[11pt]{article}
\usepackage[margin=1in]{geometry}
\usepackage{amsmath,amsthm,amssymb}
\usepackage[dvipsnames,svgnames]{xcolor}
\usepackage{hyperref}
\usepackage{cleveref}
\usepackage{dsfont}
\usepackage[linesnumbered,boxed,ruled,vlined]{algorithm2e}
\usepackage{tikz}
\usetikzlibrary{positioning}
\usepackage{hyperref}
\hypersetup{
    colorlinks=true,
    linkcolor=magenta,
    filecolor=magenta,
    citecolor=DarkRed,
    urlcolor=cyan,
}

\usepackage[T1]{fontenc}

\def\clap#1{\hbox to 0pt{\hss#1\hss}}
\providecommand{\email}[1]{\href{mailto:#1}{\texttt{#1}}}

\makeatletter
\providecommand{\qedhere}{
\global\let\zeqno\relax
\ifmmode
  \eqno \def\@badmath{$$}%$$
    \let\eqno\relax \let\leqno\relax \let\veqno\relax
    \hbox{\qed}
\else
  \qed
\fi
}
\newcommand{\setstretch}[1]{%
  \def\baselinestretch{#1}%
  \@currsize
}
\makeatother

\makeatletter
\@ifpackageloaded{tikz}{
  \tikzstyle{no sep}=[inner sep=0,outer sep=0]
}{}
\makeatother

\providecommand{\deq}{\mathrel{:=}}
\newcommand{\supp}{\operatorname{\textnormal{\text{Supp}}}}

\newcommand{\TVD}{\operatorname{\Delta_\textnormal{TV}}}

\providecommand\Ex{\mathop{\operatorname{\mathbb E}}}

\newcommand{\1}{\mathop{\mathds{1}}}

\makeatletter
\def\metadef#1#2{%
  \def\metadef@iter##1{\ifx##1;\else \expandafter\newcommand\csname#1\endcsname{#2}\expandafter\metadef@iter\fi}%
  \expandafter\metadef@iter%
}
\makeatother

\metadef{vec#1}{\mathbf #1}abcdefghijklmnopqrstuvwxyzABCD;

\metadef{ten#1}{\mathbf #1}ABCDEFGHIJKLMNOPQRSTUVWXYZ;

\metadef{frak#1}{\mathfrak #1}ABCDEFGHIJKLMNOPQRSTUVWXYZ;

\metadef{c#1}{\mathcal #1}ABCDEFGHIJKLMNOPQRSTUVWXYZ;

\metadef{s#1}{\mathsf #1}ABCDEFGHIJKLMNOPQRSTUVWXYZ;

\makeatletter
\@ifpackageloaded{mathtools}{}{
  \def\clap#1{\hbox to 0pt{\hss#1\hss}}

  \def\mathclap{\mathpalette\mathclapinternal}

  \def\mathclapinternal#1#2{\clap{$\mathsurround=0pt#1{#2}$}}
}
\makeatother

\newcommand{\bit}{\{0,1\}}
  \theoremstyle{plain}% default
  \newtheorem{theorem}{Theorem}%[section]
  \newtheorem{lemma}{Lemma}

  \theoremstyle{definition}
  \newtheorem{proposition}[lemma]{Proposition}
  \newtheorem{definition}[lemma]{Definition}%[section]
  
  \newtheorem{claim}[lemma]{Claim}

\metadef{d#1}{\mathbb #1}ABCDEFGHIJKLMNOPQRSTUVWXYZ;

\newcommand{\indpprod}{\mathbin{\cdot_{\textnormal{\sf indp}}}}
\newcommand{\NPR}{\mathcal R_\textnormal{NP}}

\newcommand{\RTVD}{\operatorname{\mathcal{R}_\textnormal{TVD}}}
\newcommand{\MTVD}{\operatorname{\Delta_\textnormal{MTV}}}
\newcommand{\LBTV}{d_{\mathrm{LB}}}

\newcommand{\epsr}{\varepsilon_{\mathrm{s}}}
\newcommand{\epsa}{\delta_{\mathrm{s}}}

\newcommand{\MK}[3][]{_{#2|#3}\if\relax#1\relax\else|{#1}\fi}

\providecommand\Dist{\mathop{\operatorname{\mathbb D}}}

\title{
On Deterministically Approximating Total Variation Distance\\
}
\author{Weiming Feng\\UC Berkeley\\\email{fwm1994@gmail.com}\and Liqiang Liu\\Peking University\\\email{lql@pku.edu.cn}\and Tianren Liu\\Peking University \\\email{trl@pku.edu.cn}}
\date{}

\begin{document}

\maketitle

\begin{abstract}
Total variation distance (TV distance) is an important measure for the difference between two distributions.
Recently, there has been progress in approximating the TV distance between product distributions:
a deterministic algorithm for a restricted class of product distributions (Bhattacharyya, Gayen, Meel, Myrisiotis, Pavan and Vinodchandran 2023)
and a randomized algorithm for general product distributions (Feng, Guo, Jerrum and Wang 2023).
We give a \emph{deterministic} fully polynomial-time approximation algorithm (FPTAS) for the TV distance between product distributions. 
Given two product distributions $\dP$ and $\dQ$ over $[q]^n$, our algorithm approximates their TV distance with relative error $\varepsilon$ in time $O\bigl( \frac{qn^2}{\varepsilon} \log q \log \frac{n}{\varepsilon \TVD(\dP,\dQ) } \bigr)$.

Our algorithm is built around two key concepts:
1) The \emph{likelihood ratio} as a distribution,
which captures sufficient information to compute the TV distance.
2) We introduce a metric between likelihood ratio distributions, called the \emph{minimum total variation distance}.
Our algorithm computes a sparsified likelihood ratio distribution that is close to the original one w.r.t.\ the new metric.
The approximated TV distance can be computed from the sparsified likelihood ratio.

Our technique also implies deterministic FPTAS for the TV distance between Markov chains.
\end{abstract}

\section{Introduction}

The \emph{total variation distance} (TV distance), which is also known as the \emph{statistical difference}, is a fundamental metric for measuring the difference between two probability distributions.
For two distributions $\dP$ and $\dQ$ over the sample space $\Omega$, their TV distance is defined by
\begin{align*}
	\TVD(\dP,\dQ) \deq \frac{1}{2}\sum_{x \in \Omega}|\dP(x) - \dQ(x)|\,. 
	%= \max_{A \subseteq \Omega}|\dP(A) - \dQ(A)|.
\end{align*}
The TV distance is essentially the $L_1$-norm of the difference, and can also be characterized as the discrepancy of the optimal coupling between $\dP$ and $\dQ$~\cite{mitzenmacher2017probability}. 
The TV distance connects to many statistical distances including the relative entropy and other well-studied divergences~\cite{PW22}.
% As a result, it is widely used in many areas of math and computer science.

% \tianren{TODO: check the discussion about SZK}

%The problem of computing TV distance was recently initiated by Bhattacharyya, Gayen, Meel, Myrisiotis, Pavan and Vinodchandran~\cite{BGMMPV22}.
This paper studies the computational problem of TV distance. 
The problem is particularly interesting when both $\dP$ and $\dQ$ are high-dimensional distributions that admit succinct descriptions.
However, computing (even approximating with additive error) the TV distance is known to be hard for general distributions.
The seminal work of Sahai and Vadhan~\cite{SahaiV03} proved that deciding whether two distributions given by Boolean circuits are ``close'' or ``far-apart'' w.r.t.\ TV distance is a complete problem in \textbf{SZK} (Statistical Zero Knowledge). 
\textbf{SZK}-complete problems are often believed to be computationally hard~\cite{BoulandCHTV20}.
Many works studied the complexity of computing the TV distance for specific classes of distributions~\cite{LP02,CMR07,Kiefer18,BGMMPV22}.
%Recently, Bhattacharyya, Gayen, Meel, Myrisiotis, Pavan and Vinodchandran~\cite{BGMMPV22} studied this problem for some special distributions.
%
One of the simplest situations is when $\dP$ and $\dQ$ are both \emph{product distributions} over the sample space $[q]^n$, where $[q] \deq \{1,2,\ldots,q\}$ is a finite domain and $n$ is the dimension.
%
% Then $\dP$ and $\dQ$ can be described by marginal distributions on every index.
%
Recently, Bhattacharyya, Gayen, Meel, Myrisiotis, Pavan and Vinodchandran~\cite{BGMMPV22} proved a surprising result: the exact computing of the TV distance between two product distributions is $\#\mathbf{P}$-complete, even in the Boolean domain ($q = 2$).
%However, the hardness result in~\cite{BGMMPV22} shows that even in the Boolean case ($q = 2$), the exact computing of the TV distance between two product distributions is $\#\mathbf{P}$-complete.

%In the same paper~\cite{BGMMPV22},  Bhattacharyya \emph{et. al.} also proposed 
%
The hardness result motivates the study of approximation algorithms for the total variation distance between two product distributions. 
The algorithm is required to estimate the TV distance within $\varepsilon$ \emph{relative error}.
Previous works mainly focused on randomized approximation algorithms, which find a good estimation with probability at least $1 - \delta$.
The first polynomial-time randomized approximation algorithm was given in the preprint version of~\cite{BGMMPV22}, but it only works for a restricted class of product distributions.
Later on, Feng, Guo, Jerrum and Wang~\cite{FGJW23} gave a randomized algorithm for general product distributions. The algorithm runs in time $O(\frac{qn^2}{\varepsilon^2}\log \frac{1}{\delta})$.

All algorithms mentioned above are based on the Monte Carlo method. One natural question is how to design deterministic approximation algorithms.
Very recently, in the conference version\footnote{Link to the paper: \url{https://www.cs.toronto.edu/~meel/Papers/ijcai23-bggmpv.pdf}} of~\cite{BGMMPV22}, authors strengthened their randomized algorithm into a deterministic one.
However, the algorithm still requires some restrictions on the input distributions, i.e., the domain needs to be Boolean and $\dQ$ has a constant number of distinct marginals (e.g., $\dQ$ can be the uniform distribution).
The major open question is to give a deterministic approximation algorithm that works for \emph{general} product distributions.
We answered this question affirmatively.

\begin{quote}
	\textbf{Theorem~\ref{thm:product}} (TV distance between product distributions)
	\it
	There exists a deterministic algorithm such that given two product distributions $\dP,\dQ$ over $[q]^n$ and $\varepsilon > 0$, it outputs $\widehat{\Delta}$ satisfying $(1-\varepsilon)\TVD(\dP, \dQ) \leq \widehat{\Delta} \leq \TVD(\dP, \dQ)$ in time $O\bigl( \frac{qn^2}{\varepsilon} \log q \log \frac{n}{\varepsilon \TVD(\dP,\dQ) } \bigr)$.
\end{quote}

In the above theorem, both $\dP$ and $\dQ$ are described by their marginal distributions and each marginal distribution is given in binary. 
We remark that the term $\log \frac{1}{\TVD(\dP,\dQ)}$ in running time is always a polynomial with respect to the input size. 
Moreover, our algorithm is faster for product distributions with larger total variation distances.
The approximated TV distance $\widehat{\Delta}$ returned by our algorithm is always a lower bound of $\TVD(\dP,\dQ)$.
%The output $\widehat{\Delta}$ is a lower bound of $\TVD(\dP,\dQ)$, and it approximates the true value within $\varepsilon$ relative error.

%which is also the number of bits to 
%The output $\widehat{\Delta}$ is a lower bound of $\TVD(\dP,\dQ)$, and it  
%and it approximates $\TVD(\dP,\dQ)$  within $\epsilon$ relative error.
%achieves the approximation
%The output $\widehat{\Delta}$ approximates $\TVD(\dP,\dQ)$ within $\epsilon$ relative error and it further satisfies $\widehat{\Delta} \leq \TVD(\dP,\dQ)$.

Our technique can be extended to more general high-dimensional distributions with ``limited'' dependency. 
One example is the distribution of the Markov chain trajectory.
Let $\dP$ be the joint distribution of a random vector $(X_1,X_2,\ldots,X_n) \in [q]^n$, where $X_1$ follows the initial distribution $\dP_1$ and each $X_i$ is generated based on a Markov kernel $\dP_{i|i-1}$ that specifies the distribution of $X_i$ conditional on $X_{i-1}$. 
We call $\dP$ the distribution of an $n$-step Markov chain\footnote{The Markov chain has the state space $[q]$ as each $X_i \in [q]$, but $\dP$ is a distribution over the sample space $[q]^n$.}.
% Given initial distributions and Markov kernels of two $n$-step Markov chains $\dP$ and $\dQ$, 
Given two $n$-step Markov chains $\dP$ and $\dQ$, 
we can also approximate their TV distance in polynomial time.
%Consider a joint distribution $(X_1,X_2,\ldots,X_n)$ 
\begin{quote}
	\textbf{Theorem~\ref{thm:markov}} (TV distance between $n$-step Markov chains)
	\it
	There exists a deterministic algorithm such that given two $n$-step Markov chains $\dP,\dQ$ in state space $[q]$ and $\varepsilon > 0$, it outputs $\widehat{\Delta}$ satisfying $(1-\varepsilon)\TVD(\dP, \dQ) \leq \widehat{\Delta} \leq \TVD(\dP, \dQ)$ in time $O\bigl( \frac{q^2n^2}{\varepsilon} \log q \log \frac{n}{\varepsilon \TVD(\dP,\dQ) } \bigr)$.
\end{quote}

We remark that our algorithm estimates the TV distance between the joint distributions of $(X_1,X_2,\ldots,X_n)$ and 
$(Y_1,Y_2,\ldots,Y_n)$ rather than the TV distance between the distributions of $X_n$ and $Y_n$.
The latter is trivial because the marginal distribution of $X_i$ (also $Y_i$) can be computed by simple matrix multiplications.
% Computing $\TVD(X_n,Y_n)$ can be solved by simple matrix multiplications, 
But estimating $\TVD(\dP, \dQ)$ is non-trivial because the sample space is exponentially large.
Our result is related to the problem of comparing two labelled Markov chains. See \Cref{sec-related} for a detailed discussion.

\subsection{Technical Overview}
\label{sec:tech-overview}
\label{sec:overview}

Let $\dP,\dQ$ be two distributions over a sample space $\Omega$,
% In the literature of binary hypothesis testing,
their total variation distance is well known to equal the maximal \emph{advantage} of distinguishing between $\dP$ and $\dQ$.
Say a distinguisher is given a sample $x$ from either $\dP$ or $\dQ$
and is asked to guess which distribution $x$ is sampled from.
The distinguisher $D$ can be formalized as a randomized algorithm (or a Markov kernel) from $\Omega$ to $\bit$.
The probability that the distinguisher $D$ outputs 0 under hypothesis $\dP$ (resp.~$\dQ$) 
can be written as $\Pr_{X\sim \dP} [D(X) \to 0]$ (resp.~$\Pr_{X\sim \dQ} [D(X) \to 0]$).
The difference between them is called the advantage.
The maximal advantage equals the total variation distance
\[ 
  \TVD(\dP,\dQ) 
  = \max_{D} \Bigl( \Pr_{X\sim \dP} [D(X) \to 0] - \Pr_{X\sim \dQ} [D(X) \to 0] \Bigr) \,,
\]
and the optimal distinguisher who maximizes the advantage, 
is the \emph{likelihood ratio test}
\[
  D(x)
  = \begin{cases}
    0, &\text{ if } \frac{\dP(x)}{\dQ(x)} > 1 \\
    1, &\text{ if } \frac{\dP(x)}{\dQ(x)} \leq 1 
  \end{cases}
\]
where $\frac{\dP(x)}{\dQ(x)}$ is called the likelihood ratio.
The likelihood ratio $\frac{\dP(x)}{\dQ(x)}$ is an important tool in information theory~\cite{PW22}.
It is particularly useful when $x$ is sampled from $\dP$ or from $\dQ$.

% Given two distributions $\dP,\dQ$ over a sample space $\Omega$,
% people call $\frac{\dP(x)}{\dQ(x)}$ the \emph{likelihood ratio} for each $x\in\Omega$.
In this paper, we formally define the likelihood ratio as a distribution, denoted by $(\dP\|\dQ)$,
such that $(\dP\|\dQ)(r) \deq \Pr_{X\sim\dQ}[\frac{\dP(X)}{\dQ(X)} = r]$.
In other words, $(\dP\|\dQ)$ is the distribution of $\frac{\dP(X)}{\dQ(X)}$ where $X\sim \dQ$.

As we will discussed in Section~\ref{sec:ratio}, 
the likelihood ratio distribution (or ``ratio'' in short) $\dR = (\dP\|\dQ)$ contains all the ``useful'' information about $(\dP,\dQ)$.
For example, $\TVD(\dP,\dQ) = \Ex_{R\sim\dR} \max(1-R, 0)$,
so we can denote the distance by $\TVD(\dR)$.
% Based on this observation, we can define a convenient notation $\TVD(\dR) \deq \TVD(\dP,\dQ)$.
%
For the task of computing the total variation distance, it suffices to compute $\dR = (\dP\|\dQ)$.

Given two product distributions $\dP = \dP_1 \dP_2 \dots \dP_n$ and $\dQ = \dQ_1 \dQ_2 \dots \dQ_n$,
their ratio $(\dP\|\dQ)$
\[
% \begin{aligned}  
  (\dP\|\dQ)(r)
  = \Pr_{(X_1,\dots,X_n)\sim\dQ} \Bigl[ \frac{\dP(X_1,\dots,X_n)}{\dQ(X_1,\dots,X_n)} = r \Bigr] 
  = \Pr_{\substack{X_1\sim\dQ_1\\\cdots\\X_n\sim\dQ_n}} \Bigl[ \underbrace{\frac{\dP_1(X_1)}{\dQ_1(X_1)}}_{R_1} \dots \underbrace{\frac{\dP_n(X_n)}{\dQ_n(X_n)}}_{R_n} = r \Bigr] 
  % = \Pr_{\substack{R_1\sim\dR_1\\\cdots\\R_n\sim\dR_n}} \Bigl[ R_1 R_2 \dots R_n = r \Bigr]
% \end{aligned}
\] 
is the distribution of $R_1R_2\dots R_n$,
where $R_1,\dots,R_n$ are independent and $R_i \sim \dR_i = (\dP_i\|\dQ_i)$.
This suggests a na\"ive algorithm to compute the total variation distance.
\begin{itemize}
\setlength\itemsep{0pt}
\setlength{\parskip}{2pt}
  \item 
    Compute $\dR_i = (\dP_i\|\dQ_i)$ for each $i\in[n]$.
  \item
    Compute $\dR = \dR_1 \indpprod \dR_2 \indpprod \cdots\dR_n$.
  \item 
    Output $\TVD(\dR)$.
\end{itemize}
Here $\dR_1 \indpprod \dR_2$ denotes the distribution of $R_1R_2$ where $R_1,R_2$ are independently sampled from $\dR_1,\dR_2$.
The na\"ive algorithm computes the \emph{exact} value of total variation distance, but has exponential time and space complexity,
because the support of $\dR$ can be exponentially large.

To contain the complexity blow-up, 
our actual algorithm computes an approximation of $\dR$.
The high-level framework of our algorithm looks as follows.
\begin{itemize}
\setlength\itemsep{0pt}
\setlength{\parskip}{2pt}
  \item 
    Compute $\dR_i = (\dP_i\|\dQ_i)$ for each $i\in[n]$.
  \item
    Compute $\dR_1\indpprod\dR_2$.
    Then sparsify it as $\tilde\dR_{1:2} \approx \dR_1\indpprod\dR_2$.
  \item
    Compute $\tilde\dR_{1:2}\indpprod\dR_3$.
    Then sparsify it as $\tilde\dR_{1:3} \approx \tilde\dR_{1:2}\indpprod\dR_3$.
  \item[]\phantom{oo}\smash{$\vdots$}
  \item
    Compute $\tilde\dR_{1:n-1}\indpprod\dR_n$.
    Then sparsify it as $\tilde\dR_{1:n} \approx \tilde\dR_{1:n-1}\indpprod\dR_n$.
  \item 
    Output $\TVD(\dR_{1:n})$ as an approximation of $\TVD(\dP,\dQ)$.
\end{itemize}
In our algorithm, the \emph{sparsification} subroutine plays a central role:
a given complicated ratio is sparsified into a simpler one that is ``close to'' the given ratio.
The closeness is formalized as a new metric.
For the correctness of our algorithm,
the metric should satisfy a few properties:
i)~triangle inequality, if $\dR\approx \dR'$ and $\dR'\approx\dR''$, then $\dR\approx\dR''$;
ii)~if $\dR\approx \tilde\dR$, then $\TVD(\dR) \approx \TVD(\tilde\dR)$;
iii)~if $\dR \approx \tilde\dR$, then $\dR\indpprod\dR' \approx \tilde\dR\indpprod\dR'$.
Then correctness is guaranteed as follows: 
% $\tilde\dR_{1:i+1} \approx \tilde\dR_{1:i}\indpprod\dR_{i+1}$ implies
\[ 
  \tilde\dR_{1:i+1} \approx \tilde\dR_{1:i}\indpprod\dR_{i+1}
  \implies
  \tilde\dR_{1:i+1} \indpprod\dR_{i+2} \indpprod\cdots\dR_{n} \approx \tilde\dR_{1:i}\indpprod\dR_{i+1} \indpprod\dR_{i+2} \indpprod\cdots\dR_{n} \,;
\]
then by triangle inequality, $\tilde\dR_{1:n} \approx \dR_1 \indpprod \dR_2 \indpprod \cdots\dR_n = \dR$;
therefore $\TVD(\dR) \approx \TVD(\tilde\dR)$.

We say two ratios $\dR,\dR'$ are close,
if there exist distributions $\dP,\dQ,\dP',\dQ'$ satisfying $(\dP\|\dQ)=\dR$, $(\dP'\|\dQ')=\dR'$ such that $\dP$ (resp.~$\dQ$) is close to $\dP'$ (resp~$\dQ'$) with respect to the TV distance.
In such case, $\TVD(\dR)$ ($= \TVD(\dP,\dQ)$) must be close to $\TVD(\dR')$ ($= \TVD(\dP',\dQ')$) due to the triangle inequality of the TV distance.
This inspires us to consider the \emph{minimum total variation distance}, 
\[
  \MTVD(\dR,\dR')
  = \min_{\substack{(\dP\|\dQ)=\dR \\ (\dP'\|\dQ')=\dR'}} \max\bigl( \TVD(\dP,\dP'), \TVD(\dQ,\dQ') \bigr) \,.
\]
In Section~\ref{sec:MTVD}, we show that $\MTVD$ is a metric between ratios.

In the rest of this overview, 
we briefly describe how the sparsification subroutine works,
and why the sparsified ratio is close to the original ratio with respect to the minimum total variation distance.

% \tianren{$\downarrow$ replace $p_i$ by $q_i$?}
% \tianren{$q$ have been used to denote the domain size}

The sparsification subroutine takes a ratio $\dR$ as the input.
The ratio $\dR$ can be represented by a table of entries $(r_1,p_1),(r_2,p_2),\dots$,
each entry $(r_i,p_i)$ represents $\dR(r_i) = p_i$.
If the table has too many entries, then some of them must be close together.
Say $(r_1,p_1), \ldots, (r_k,p_k)$ are very close together in the sense that %$r_1\approx r_2\approx r_3$.
% W.l.o.g., assume that 
$r_1,\dots,r_k \in [a,b)$, where $[a,b)$ is a sufficiently narrow interval.
W.l.o.g., assume $[a,b) \subseteq [0,1]$, and the case inside $[1,\infty]$ is symmetric.
The sparsification subroutine ``merges'' these $k$ entries into one entry.
% into one entry $(\frac{r_1p_1+r_2p_2+r_3p_3}{p_1+p_2+p_3}, p_1+p_2+p_3)$.
The sparsified ratio $\tilde\dR$ is represented by table $(\frac{\sum_{i\in[k]} r_ip_i}{\sum_{i\in[k]}p_i}, \sum_{i\in[k]}p_i), (r_{k+1},p_{k+1}),\allowbreak (r_{k+2},p_{k+2}), \dots$
% \[
%   \Bigl(\frac{r_1p_1+r_2p_2+r_3p_3}{p_1+p_2+p_3}, p_1+p_2+p_3\Bigr), (r_4,p_4), (r_5,p_5), \dots
% \]
We claim that $\tilde\dR$ is close to $\dR$ with respect to the minimum total variation distance $\MTVD$.

Define two distributions $\dP,\dQ$ over sample space $\Omega = \{\omega_0,\omega_1,\omega_2,\dots\}$ as
\[
  \dP(\omega_i) = r_i p_i, \quad
  \dQ(\omega_i) = p_i, \quad
  \dP(\omega_0) = 1-\Ex_{R\sim\dR}[R], \quad
  \dQ(\omega_0) = 0,
\] 
so that $(\dP\|\dQ) = \dR$.
Consider distribution $\tilde\dP$ that is the same as $\dP$ except 
\[
  \tilde\dP(\omega_i) = r^* p_i
  \text{ for $i\in[k]$, where } r^* \deq \frac{\sum_{i\in[k]} r_ip_i}{\sum_{i\in[k]}p_i} \,.
\]
Then $\tilde\dR = (\tilde\dP \| \dQ)$, this also ensures that $\tilde\dR$ is a valid ratio.
By definition, $\MTVD(\dR,\tilde\dR) \leq \TVD(\dP,\tilde\dP) \leq \sum_{i\in[k]} |r_i-r^*|\, p_i$.
Since $r_1,\dots,r_k,r^* \in [a,b]$, 
we have $|r_i-r^*|$ being small for all $i\in[k]$.
Quantitatively, say interval $[a,b]$ is sufficiently narrow means $\frac{1-a}{1-b} \leq 1+\varepsilon$.
Thus we have $|r_i-r^*| \leq \varepsilon (1-r_i)$ for all $i\in[k]$,
and the error introduced by merging is bounded by $\MTVD(\dR,\tilde\dR) \leq \sum_{i\in[k]} \varepsilon (1-r_i)p_i$.

The actual sparsification subroutine picks a collection of disjointed sufficiently narrow intervals\footnote{%
The last interval is \emph{not} sufficiently narrow because its right endpoint is 1,
and has to be analyzed separately.}
$[a_1,b_1),[a_2,b_2),\dots,[a_m,b_m)$ that jointly cover $[0,1)$,
then merges all the entries that lie in the same interval.
The total error introduced by merging is bounded by 
\[
  \sum_{j=1}^{m-1}
  \sum_{i \text{ s.t.~} r_i \in [a_j,b_j]} \varepsilon (1-r_i)p_i
  \leq
  \sum_{i \text{ s.t.~} r_i < 1} \varepsilon (1-r_i)p_i
  = 
  \varepsilon \Ex_{R\sim\dR} [\max(1-R,0)]
  =
  \varepsilon \TVD(\dR) \,.
\]
Symmetrically, the error introduced by merging within $(1,\infty]$ is also bounded by $\varepsilon \TVD(\dR)$.

\subsection{Related Works and Open Problems}\label{sec-related}

The TV distance between two labelled Markov chains (LMC, a.k.a. hidden Markov chain) has been studied by previous works~\cite{LP02,CMR07,Kiefer18}. It was proved that for LMCs, computing TV distance within $\varepsilon$-additive error, where $\varepsilon > 0$ is given in binary, is \#\textbf{P}-hard~\cite{Kiefer18}.
The Markov chains we studied in \Cref{thm-main-2} is \emph{not hidden}, which is equivalent (by a polynomial-time reduction) to the \emph{deterministic acyclic LMCs} in \cite{Kiefer18}.
Kiefer~\cite[Theorem 10]{Kiefer18} gave a randomized algorithm that approximates the TV distance between  (not necessarily deterministic) acyclic LCMs with \emph{additive-error} $\epsilon$, where the algorithm succeeds with probability at least $1-\delta$ and the running time is $\mathrm{poly}(\frac{1}{\epsilon},\log \frac{1}{\delta}, \text{input size})$.
Compared with our result, Kiefer's algorithm works for more general distributions but for deterministic acyclic LMCs, our algorithm is deterministic and achieves stronger relative-error approximation.
%Our algorithm approximates TV distance within relative error instead of additive error. 

%For this special class of LCMs, our algorithm achieves better accuracy than the algorithm in \cite[Theorem 10]{Kiefer18}.
%However, the algorithmic result in~\cite{Kiefer18} works for general (not necessarily deterministic) acyclic LCMs.

Some works~\cite{CDKS18,BGMV20} studied the problem of computing TV distances for structured high-dimensional distributions, e.g. Bayesian networks, Ising models and multivariate Gaussian distributions. 
%
%Some approximation algorithms were discovered but they are randomized and only guarantee the additive approximation error.
%
Some randomized approximation algorithms were discovered, but they only achieved an additive-error approximation. %which is weaker than the relative-error approximation.
One open problem is to find efficient deterministic approximation (with additive or relative error) algorithms for those distributions.
Our technique relies on a strong conditional independence property of the distribution.
We wonder whether one can relax this restriction and make our technique work for more general graphical models.

Our technique for approximating the TV distance is different from the previous ones.
Previous randomized algorithms~\cite{Kiefer18,CDKS18,BGMV20,FGJW23} are all based on the Monte Carlo method.
For the deterministic approximation of the TV distance between two product distributions, the algorithm in~\cite{BGMMPV22} first reduces the problem to the approximate counting of knapsack solutions (\#Knapsack) and then solves the \#Knapsack by existing deterministic approximation algorithms~\cite{gopalan2010polynomial,StefankovicVV12}.
We introduce a new sparsification technique together with a new metric called minimum total variation distance, which is of independent interest.

%For approximating the TV distance between product distributions, our technique is very different from the previous ones. 
%
%The algorithm in \cite{FGJW23} utilized the \emph{coupling technique} to design an unbiased estimator and use the Monte Carlo method to approximate the TV distance. .......
%
%Both randomized and deterministic algorithms in \cite{BGMMPV22} first reduce the original problem to the problem of approximating the number of knapsack solutions 

%Previous randomized algorithms~\cite{Kiefer18,CDKS18,BGMV20,FGJW23} all use the Monte Carlo method and the algorithm in \cite{FGJW23} further utilise the \emph{coupling technique} to approximate the TV distance between production

%Another related problem is to approximate the TV-distance with additive error.
%Although the problem is hard for general distributions~\cite{SahaiV03},  there are some randomized algorithms for special classes of distributions including Bayesian networks, undirected graphical models and multivariate Gaussian distributions~\cite{CDKS18,BGMV20}.
%
 %even for the addition approximation error.
%

%{\color{blue}Say something about testing and distinguishing ? seems related to our technique}

\section{Notations}

% no need to polish this section at the moment
We use $[n]$ to denote the set $\{1,2,\ldots,n\}$ for any positive integer $n$.
We use $\1[A] \in \{0,1\}$ to indicate whether the condition (or event) $A$ holds.

We use $\infty$ to denote the infinity point, which is an extended real number.
For any $a\in(0,\infty)$, let ${a}/0 = \infty$, $a/\infty = 0$ and $a<\infty$.

This paper only considers \emph{discrete} distributions.
Distributions are denoted by calligraphic letters (e.g.,~$\dP,\dQ$).
A {discrete} distribution $\dP$ over a sample space $\Omega$ can be defined by its probability mass function $\dP: \Omega \to [0,1]$.
For each $x\in\Omega$, the value $\dP(x)$ is the probability that $x$ is sampled from $\dP$.
We use $\supp(\dP)$ to denote the \emph{support} of $\dP$, which refers to $\{x \in \Omega \mid \dP(x) > 0\}$.
Let $X \sim \dP$ denote that \emph{random variable} $X$ follows distribution $\dP$.
For each subset $S \subseteq \Omega$, we use the conventional notation $\dP(S)$ to denote the probability that a sampling from $\dP$ falls in set $S$.  
That is, $\dP(S) \deq \Pr_{X\gets \dP} [X \in S]$.

For any two distributions $\dP,\dQ$, 
let $\dP \dQ$ (or $\dP \times \dQ$) denotes the distribution of $(X, Y)$,
let $\dP \indpprod \dQ$ denotes the distribution of $X\cdot Y$,
where $X,Y$ are independent random variables satisfying $\dP,\dQ$ respectively.

A Markov kernel $\kappa$ from sample space $\Omega$ to sample space $\Omega'$ is a function $\kappa:\Omega'\times\Omega\to[0,1]$,
such that for any $x\in\Omega$, function $y\mapsto \kappa( y | x)$ is a distribution over $\Omega'$.
For any distribution $\dP$ over $\Omega$,
% let $(\dP,\kappa)$ denote the distribution over $\Omega\times\Omega'$ that $(\dP,\kappa)(x,y) = \dP(x) \kappa(y|x)$;
let $\kappa\dP$ denote the distribution over $\Omega'$ that $(\kappa\dP)(y) = \sum_{x\in\Omega} \dP(x) \kappa(y|x)$.
Their semidirect product, denoted by $\dP\kappa$ or $\dP\times\kappa$,
is a distribution over $\Omega\times\Omega'$ that $(\dP\kappa)(x,y) = \dP(x) \kappa(y|x)$.

\section{Likelihood Ratio as a Distribution}
\label{sec:ratio}

Given two distributions $\dP,\dQ$ over a sample space $\Omega$,
people call $\frac{\dP(x)}{\dQ(x)}$ the \emph{likelihood ratio} for $x\in\Omega$.
We introduce a notation for the distribution of the likelihood ratio.

\begin{definition}[Ratio]\label{def:ratio}
  Let $\dP,\dQ$ be two discrete distributions over a sample space $\Omega$.
  The likelihood ratio distribution (or \emph{ratio}, in short) between $\dP,\dQ$, denoted by $(\dP\|\dQ)$, is a distribution over $[0,\infty)$.
  % In the paper, we typically define $\dR \deq (\dP\|\dQ)$
  % such that
  % Such that $R = (P\|Q)$ if and only if for all $r\in[0,\infty)$
  \[
    (\dP\|\dQ)(r) \deq \Pr_{X\sim \dQ} \Bigl[ \frac{\dP(X)}{\dQ(X)} = r \Bigr].
  \]
\end{definition}

Equivalently, the distribution $(\dP\|\dQ)$ can be defined as follows: 
sample random variable $X \sim \dQ$,
define $(\dP\|\dQ)$ as the distribution of $\frac{\dP(X)}{\dQ(X)}$.
Note that the denominator is always non-zero because $X$ is sampled from $\dQ$.
In the paper, we typically denote the ratio distribution by $\dR = (\dP\|\dQ)$.

Not every discrete distribution $\dR$ over $[0,\infty)$ is the ratio between some two discrete distributions.
We say $\dR$ is a \emph{valid ratio} if there exist two discrete distributions $\dP,\dQ$ such that $\dR = (\dP \| \dQ)$.

\begin{proposition}%[Condition of Valid Ratio Random Variable]
\label{prop-valid}
A discrete distribution $\dR$ over $[0,\infty)$ is a valid ratio if and only if $\Ex_{R\sim \dR}[R] \leq 1$.
\end{proposition}

\begin{proof}%[Proof of \Cref{prop-valid}]
Suppose $\dR = (\dP\|\dQ)$,
then
$\Ex_{R\sim \dR}[R]
= \sum_{x \in \supp(\dQ)} \dQ(x) \frac{\dP(x)}{\dQ(x)} 
= \dP( \supp(\dQ) )
\leq 1$.

On the other hand, 
if $\Ex_{R\sim\dR}[R] \leq 1$,
define distribution $\dR^\dagger$ over $[0,\infty] = [0,\infty) \cup \{\infty\}$ as
\begin{equation}
\label{eq:canonical}
  \dR^\dagger(r) = \begin{cases}
    r \cdot \dR(r), &\text{ if } r \in [0,\infty) \\
    1 - \Ex_{R\sim\dR}[R], &\text{ if } r = \infty.
  \end{cases} 
\end{equation}
It is easy to verify that $\dR^\dagger$ is a distribution and $\dR = (\dR^\dagger\|\dR)$.
\end{proof}

\begin{definition}[Canonical Pair]\label{def-can-pair}
For any valid ratio $\dR$, 
define its alternative ratio, denoted by $\dR^\dagger$, as the distribution in~\eqref{eq:canonical}. 
We call $(\dR^\dagger,\dR)$ the \emph{canonical pair} of $\dR$.
\end{definition}

We can define another natural distribution of likelihood ratio $(\dP\|\dQ)^\dagger$
% when $X$ is sampled from $\dP$.
% Define distribution $(\dP\|\dQ)^\dagger$ as 
\[
  (\dP\|\dQ)^\dagger(r) \deq \Pr_{X\sim \dP} \Bigl[ \frac{\dP(X)}{\dQ(X)} = r \Bigr].
\]
Its sample space has to be extended to $[0,\infty]$, because the denominator can be zero.
It is easy to verify that $\dR = (\dP\|\dQ)$ implies $\dR^\dagger = (\dP\|\dQ)^\dagger$.

\begin{proposition}\label{prop:product}\label{prop-product}
  $(\dP_1 \dP_2\|\dQ_1 \dQ_2) = (\dP_1\|\dQ_1) \indpprod (\dP_2\|\dQ_2)$
  for any discrete distributions $\dP_1,\dQ_1,\dP_2,\dQ_2$.
\end{proposition}

\begin{proof}
  Let $(X_1,X_2) \sim \dQ_1\dQ_2$. 
  Define $R_i \deq \frac{\dP_i(X_i)}{\dQ_i(X_i)}$, then $R_i \sim (\dP_i \| \dQ_i)$.
  Since $X_1,X_2$ are independent, $R_1,R_2$ are also independent.
  Define $R \deq R_1R_2 = \frac{(\dP_1\dP_2)(X_1,X_2)}{(\dQ_1\dQ_2)(X_1,X_2)}$, then $R \sim (\dP_1 \dP_2\|\dQ_1 \dQ_2)$.
\end{proof}

The likelihood ratio is a well-studied concept.
Almost all statistical distances and divergences between $\dP$ and $\dQ$ can be derived from their ratio $\dR = (\dP\|\dQ)$.
The derivation of the total variation distance is shown in the following proposition.
The cases of other distances and divergences can be found in most information theory textbooks (e.g.~\cite{PW22}).

\begin{proposition}%[Ratio Random Variable and TV Distance]\label{prop-ratio-tvd}
  Let $\dR = (\dP\|\dQ)$, then $\TVD(\dP,\dQ) = \TVD(\dR)$, where
  \[
    \TVD(\dR) \deq  \Ex_{R\sim\dR}\bigl[ (1-R)\cdot \1[R < 1] \bigr] =\Ex_{R\sim\dR}\bigl[ \max(1-R,0)\bigr].
  \]
\end{proposition}

\begin{proof}
  $\displaystyle
    \TVD(\dP,\dQ) 
    % &= \frac{1}{2} \sum_{x \in \Omega} \vert P(x) - Q(x) \vert 
    = \sum_{\substack{x \text{~s.t.}\\\mathclap{\dQ(x) > \dP(x)}}} \Bigl( \dQ(x) - \dP(x) \Bigr)
    = \sum_{\substack{x \text{~s.t.}\\\mathclap{\dQ(x) > \dP(x)}}} \dQ(x) \Bigl(  1-\frac{\dP(x)}{\dQ(x)} \Bigr)
    = \Ex_{R\sim \dR}\Bigl[ (1-R)\cdot \1[R < 1] \Bigr]. 
  $
  % \begin{equation*}
  %   \TVD(\dP,\dQ) 
  %   % &= \frac{1}{2} \sum_{x \in \Omega} \vert P(x) - Q(x) \vert 
  %   = \sum_{\substack{x \text{~s.t.}\\\mathclap{\dQ(x) > \dP(x)}}} \Bigl( \dQ(x) - \dP(x) \Bigr)
  %   = \sum_{\substack{x \text{~s.t.}\\\mathclap{\dQ(x) > \dP(x)}}} \dQ(x) \Bigl(  1-\frac{\dP(x)}{\dQ(x)} \Bigr)
  %   = \Ex_{R\sim \dR}\Bigl[ (1-R)\cdot \1[R < 1] \Bigr]. \qedhere
  % \end{equation*}
\end{proof}

It comes as no surprise how ratio can derive these distances and divergences. 
In some sense, ratio $\dR = (\dP\|\dQ)$ is a complete characterization of the decision problem between $\dP$ and $\dQ$.
% We provides two perspectives arguing the completeness.

\begin{description}
  \item[Sufficient Statistic.]
    Say random variable $X$ is sampled from either $\dP$ or $\dQ$.
    Whether $X$ is sampled from $\dP$ or $\dQ$ is the hidden parameter of the statistical model.
    Consider a statistic $R = \kappa(X)$ where $\kappa(x) \deq \frac{\dP(x)}{\dQ(x)}$.
    If the hidden parameter is $\dQ$, then $R\sim \dR$;
    otherwise, $R\sim \dR^\dagger$.
    That is, $\dR = \kappa\dQ$, $\dR^\dagger = \kappa\dP$.
    It is well known that $R$ is a \emph{sufficient statistic}
    containing all the ``useful'' information of $X$.
    There exists a Markov kernel $\kappa^{-1}$ who recovers 
    the entire information of $(\dP,\dQ)$ from $(\dR^\dagger,\dR)$,
    such that $\dQ = \kappa^{-1}\dR$, $\dP = \kappa^{-1}\dR^\dagger$.
    So the problem of distinguishing $\dP,\dQ$
    is \emph{equivalent} to the problem of distinguishing $\dR^\dagger,\dR$.

  \item[Equivalent Decision Problems.]
    In decision theory, 
    the equivalence relationship between decision problems are formalized.
    Two decision problems $(\dP,\dQ)$ and $(\dP',\dQ')$ are called equivalent, 
    if there exist Markov kernels $\kappa,\kappa^{-1}$ such that $\kappa\dP=\dP', \kappa\dQ=\dQ', \kappa^{-1}\dP'=\dP, \kappa^{-1}\dQ'=\dQ$~\cite{LeCam}.
    It is easy to verify that they are equivalent if and only if $(\dP\|\dQ) = (\dP'\|\dQ')$.
    If direction: let $\dR = (\dP\|\dQ)= (\dP'\|\dQ')$, then both $(\dP,\dQ)$ and $(\dP',\dQ')$ equivalent to the canonical pair $(\dR^\dagger,\dR)$.
    Only if direction is ensured by a new metric defined in Section~\ref{sec:MTVD}.
    Thus $\dR = (\dP\|\dQ)$ can represent the equivalence class where $(\dP,\dQ)$ is. 
    % The canonical pair $(\dR^\dagger,\dR)$ is also a natural representative of the equivalence class.
\end{description}

The above discussion inspires us to consider the ``data processing'' relation.
Say decision problem $(\dP,\dQ)$ is stronger than $(\dP',\dQ')$, denoted by $(\dP,\dQ) \geq (\dP',\dQ')$,
if and only if there exists a Markov kernel $\kappa$ such that $\dP' = \kappa\dQ$, $\dQ' = \kappa\dP$.
This relation is almost an order relation~\cite{LeCam},
satisfying reflexivity, transitivity and antisymmetry w.r.t.~the equivalence relation between decision problems
\[
  (\dP,\dQ) \geq (\dP',\dQ')
  % \text{ and }
  ~\wedge~
  (\dP',\dQ') \geq (\dP,\dQ)
  ~\text{ if and only if }~
  (\dP\| \dQ) = (\dP'\|\dQ') \,.
\]
Therefore, it induce an order relation between ratios. 

\begin{definition}\label{def:order}
  For any two ratios $\dR_1,\dR_2$,
  we say ``$\dR_1$ is stronger than $\dR_2$'',
  denoted by $\dR_1 \geq \dR_2$,
  or ``$\dR_2$ is weaker than $\dR_1$'',
  denoted by $\dR_2 \leq \dR_1$,
  if there exists distributions $\dP,\dQ$ and a Markov kernel $\kappa$
  such that $\dR_1 = (\dP\|\dQ)$ and $\dR_2 = (\kappa\dP \| \kappa\dQ)$.
\end{definition}

% As the notation suggested, it is a order relation (the proof is deferred to Section~\ref{}).
It captures the ``data processing'':
$\kappa$ is the process,
$\kappa\dP, \kappa\dQ$ are the post-processing distributions.
Therefore, by the data processing inequality for TV distance,
$\dR_1 \geq \dR_2 \implies \TVD(\dR_1) \geq \TVD(\dR_2)$.

% \begin{lemma}
%   For any valid ratios $\dR_1 \geq \dR_2$, 
%   it holds that $\TVD(\dR_1) \geq \TVD(\dR_2)$.
% \end{lemma}

\begin{proposition}\label{prop:product-order}
  For any ratios satisfying $\dR_1 \geq \dR_1'$, $\dR_2 \geq \dR_2'$,
  it holds that $\dR_1 \indpprod \dR_2  \geq \dR_1' \indpprod \dR_2'$.
\end{proposition}

\begin{proof}
  There exist distributions $\dP_1,\dQ_1,\dP_2,\dQ_2$ and Markov kernels $\kappa_1,\kappa_2$ such that
  $\dR_i = (\dP_i \| \dQ_i)$ and $\dR_i' = (\kappa_i\dP_i \| \kappa_i\dQ_i)$.
  Then $\dR_1 \indpprod \dR_2 = (\dP_1\dP_2 \| \dQ_1 \dQ_2)$
  and $\dR_1' \indpprod \dR_2' = ((\kappa_1\dP_1)(\kappa_2\dP_2) \| (\kappa_1\dP_1)(\kappa_2\dP_2))$.
  So it suffices to find a Markov kernel $\kappa$ such that
  $\kappa(\dP_1\dP_2) = (\kappa_1\dP_1)(\kappa_2\dP_2)$, $\kappa(\dQ_1\dQ_2) = (\kappa_1\dQ_1)(\kappa_2\dQ_2)$.
  The required Markov kernel is $\kappa(x_1',x_2'|x_1,x_2) \deq \kappa_1(x_1'|x_1) \kappa_2(x_2'|x_2)$.
  % Consider Markov kernel $\kappa(x_1',x_2'|x_1,x_2) \deq \kappa_1(x_1'|x_1) \kappa_2(x_2'|x_2)$,
  % then 
  % $\kappa(\dP_1\dP_2) = (\kappa_1\dP_1)(\kappa_2\dP_2)$ and $\kappa(\dQ_1\dQ_2) = (\kappa_1\dQ_1)(\kappa_2\dQ_2)$.
\end{proof}

\subsection{Minimum Total Variation Distance}
\label{sec:MTVD}

We introduce a metric between (valid) ratios, called the minimum total variation distance.
% \tianren{I will talk about LeCam \cite{LeCam}}

\begin{definition}[MTV Distance]\label{def:min-tvd}\label{def:MTVD}
  For two valid ratios $\dR_1,\dR_2$,
  the \emph{minimum total variation distance} between them, 
  denoted by $\MTVD(\dR_1,\dR_2)$,
  is defined as
  \[
    \MTVD(\dR_1,\dR_2)
    \deq \inf_{\substack{\textit{discrete}~\dP_1,\dQ_1,\dP_2,\dQ_2\\(\dP_1\|\dQ_1) = \dR_1 \\(\dP_2\|\dQ_2) = \dR_2}} 
    \max\bigl( \TVD(\dP_1,\dP_2), \TVD(\dQ_1,\dQ_2)\bigr).
  \]
\end{definition}

We define $\MTVD$ as an infimum for safe.
% The search space of $(\dP_1,\dQ_1,\dP_2,\dQ_2)$ is infinite-dimensional and non-compact.
We cannot rule out the possibility that the minimum does not exist,
unless the supports of $\dR_1$ and $\dR_2$ are finite (Lemma~\ref{lem:inf}).

% The infimum symbol in the definition can be replaced by minimum,
% when the supports of $R_1$ and $R_2$ are finite (Lemma~\ref{lem:inf}).

\begin{lemma}
\label{lem:inf}
  When the supports of ratios $\dR_1$ and $\dR_2$ are finite,
  \[
    \MTVD(\dR_1,\dR_2)
    = \min_{\substack{\textit{discrete}~\dP_1,\dQ_1,\dP_2,\dQ_2\\(\dP_1\|\dQ_1) = \dR_1 \\(\dP_2\|\dQ_2) = \dR_2}} 
    \max\bigl( \TVD(\dP_1,\dP_2), \TVD(\dQ_1,\dQ_2)\bigr).
  \]
  % \tianren{w.l.o.g.~we can assume the minimum $(P_1,Q_1,P_2,Q_2)$ also have finite supports} 
\end{lemma}

The lemma is proved by enforcing the search space of $(\dP_1,\dQ_1,\dP_2,\dQ_2)$ to be a tuple of distributions over a fixed finite domain. The proof is deferred to \Cref{sec-proof-mtvd}.

% This paper focuses on the ratio distributions with finite support.
The minimum total variance distance satisfies the following properties, as proved in \Cref{sec-proof-mtvd}.

\begin{lemma}\label{lem-metric}
  The minimum total variance distance is a metric.
\end{lemma}

\begin{lemma}\label{lem:error-MTVD}
For any two valid ratios $\dR_1,\dR_2$, it holds that $\left\vert \TVD(\dR_1) - \TVD(\dR_2) \right\vert\leq 2\MTVD(\dR_1,\dR_2)$.
% \begin{align*}
% 	\left\vert \TVD(\dR_1) - \TVD(\dR_2) \right\vert\leq \MTVD(\dR_1,\dR_2),
% \end{align*}
% which implies for any discrete distributions $\dP_1,\dQ_1,\dP_2,\dQ_2$ with $\dR_1 = (\dP_1 \| \dQ_1)$ and $\dR_2 = (\dP_2 \| \dQ_2)$, it holds that $\bigl| \TVD(\dP_1,\dQ_1) - \TVD(\dP_2,\dQ_2) \bigr| \leq \MTVD(\dR_1,\dR_2)$.
\end{lemma}
\begin{lemma}\label{lem:product-MTVD}\label{prop-cross}
For any valid ratios $\dR_1,\dR_2,\dR_3,\dR_4$,
\begin{align*}
	\MTVD(\dR_1 \indpprod \dR_2, \dR_3 \indpprod \dR_4 ) \leq \MTVD(\dR_1, \dR_3) + \MTVD(\dR_2, \dR_4). 
\end{align*}  	
\end{lemma}
% The proofs of the above three properties are given in \Cref{sec-proof-mtvd}.

A very similar concept called \emph{Le\,Cam's distance} exists in decision theory. 
% A equivalent definition is presented in Definition~\ref{def:deficiency}.
Here we present an equivalent definition restricting on the decision problems that has only two hypotheses.

\begin{definition}[Deficiency and Le Cam's Distance \cite{LeCam}]\label{def:deficiency}
  For decision problems $(\dP_1,\dQ_1), (\dP_2,\dQ_2)$,
  the \emph{deficiency} of $(\dP_1,\dQ_1)$ with respect to $(\dP_2,\dQ_2)$
  is defined as
  \[
    \text{deficiency}((\dP_1,\dQ_1), (\dP_2,\dQ_2))
    \deq 
    \inf_{\substack{\textit{Markov kernel}~\kappa}} 
    \max \bigl( \TVD(\kappa\dP_1,\dP_2), \TVD(\kappa\dQ_1,\dQ_2)\bigr) \,;
  \]
  their \emph{Le\,Cam's distance} is $\max(\text{deficiency}((\dP_1,\dQ_1), (\dP_2,\dQ_2)), \text{deficiency}((\dP_2,\dQ_2), (\dP_1,\dQ_1)))$.
\end{definition}

% By definition, $\text{deficiency}((\dP_1,\dQ_1), (\dP_2,\dQ_2)) = 0$ if and only if $(\dP_1\|\dQ_1) \geq (\dP_2\|\dQ_2)$.

Le Cam's distance is a pseudo metric between decision problems. 
It is not hard to show that two decision problems $(\dP_1,\dQ_1), (\dP_2,\dQ_2)$ have zero Le Cam's distance if and only if $(\dP_1\|\dQ_1) = (\dP_2\|\dQ_2)$.
Therefore, Le Cam's distance is a metric between ratios.
We can similarly define deficiency between ratios.
% Similarly, deficiency also can be defined between ratios. 
\begin{align*}
  \text{deficiency}(\dR_1,\dR_2)
  &\deq \inf_{\substack{\textit{discrete}~\dP_1,\dQ_1,\dP_2,\dQ_2\\(\dP_1\|\dQ_1) \leq \dR_1 \\(\dP_2\|\dQ_2) = \dR_2}} 
  \max\bigl( \TVD(\dP_1,\dP_2), \TVD(\dQ_1,\dQ_2)\bigr) \,, \\
  \text{Le Cam's distance} (\dR_1,\dR_2) &\deq \max (\text{deficiency}(\dR_1,\dR_2), \text{deficiency}(\dR_2,\dR_1)) \,.
\end{align*}

For any two ratios $\dR_1,\dR_2$,
it holds that
$\text{deficiency}(\dR_1, \dR_2) = 0$ if and only if $\dR_1 \geq \dR_2$,
and $\text{Le Cam's distance} (\dR_1,\dR_2) \leq \MTVD(\dR_1,\dR_2)$.
% \[
%   \text{Le Cam's distance} (\dR_1,\dR_2) \leq \MTVD(\dR_1,\dR_2).
% \]
%\tianren{Can we show Le Cam's distance is close to MTVD (with a constant multiplicative factor)?}
% Similarly, deficiency can be defined between ratios: $\LCd((\dP_1 \| \dQ_1), (\dP_2 \| \dQ_2)) \deq \LCd((\dP_1,\dQ_1), (\dP_2,\dQ_2))$.
% We also conjecture that 
% $\text{deficiency}(\dR_1, \dR_2) + \text{deficiency}(\dR_2, \dR_1) \geq \MTVD(\dR_1,\dR_2)$.

% \tianren{$\downarrow$ unfinished; may be moved; may be removed}
% When $\LCd(\dR_1, \dR_2) = 0$,
% we say ``$\dR_1$ is stronger than $\dR_2$'' 
% and denote it by $\dR_1 \geq \dR_2$.
% As the notation suggested, it is a order relation.
% This order relation can also be directly defined without basing on deficiency.

% \begin{definition}
%   For any two ratios $\dR_1,\dR_2$,
%   if there exists distribution $\dP,\dQ$ and a Markov kernel $\kappa$
%   such that $\dR_1 = (\dP\|\dQ)$ and $\dR_2 = (\kappa\dP \| \kappa\dQ)$,
%   we say ``$\dR_1$ is stronger than $\dR_2$'',
%    denoted by $\dR_1 \geq \dR_2$.
% \end{definition}

% \begin{proposition}
%   \begin{enumerate}
%     \item 
%       Reflexive:
%       $\dR \geq \dR$.
%     % \item 
%     %   If $(\dP_1\|\dQ_1) \geq (\dP_2\|\dQ_2)$,
%     %   there exists a Markov kernel $\kappa$
%     %   such that  $\kappa\dP_1 = \dP_2$ and $\kappa\dQ_1 = \dQ_2$.
%     \item
%       Antisymmetric:
%       $\dR_1 \geq \dR_2$ and $\dR_2 \geq \dR_1$ 
%       implies $\dR_1 = \dR_2$ 
%     \item
%       Transitive:
%       $\dR_1 \geq \dR_2$ and $\dR_2 \geq \dR_3$ 
%       implies $\dR_1 \geq \dR_3$
%   \end{enumerate}
% \end{proposition}

\section{Sparsify the Likelihood Ratio}

%\tianren{define and analyze the sparsification algorithm in this section}
As discussed in Section~\ref{sec:ratio},
the ratio $\dR = (\dP\|\dQ)$ completely characterizes the problem of distinguishing $\dP$ and $\dQ$.
Once the ratio is known, the distances between $\dP,\dQ$ (including the TV distance) follow easily.
The bottleneck is the complexity of computing and representing $\dR$.
We will store ratio $\dR$ as the table of values of its probability mass function,
so the space complexity is proportional to the support size $|\supp(\dR)|$.
When $\dP,\dQ$ are described as product distributions,
the size of $\supp(\dR)$ can be exponentially large.
Our solution is to simplify the ratio $\dR$ without introducing too much error.
The process of simplification is called the \emph{sparsification} of the ratio.
The amount as error is measured by MTV distance (Definition~\ref{def:MTVD}).

\begin{lemma}[Sparsification]\label{thm:spar}\label{lem:sparsify}
  There exists a deterministic algorithm \textsf{Sparsify} as defined in~\eqref{eq:def-sparsify}.
  Given a valid ratio $\dR$ and two error bounds $\epsr, \epsa > 0$ as inputs,
  it outputs a sparsified ratio $\tilde\dR \leq \dR$ in time $O(\frac{1}{\epsr}\log \frac{1}{\epsa} + |\supp(\dR)|)$,
  such that $|\supp(\tilde\dR)|= O(\frac{1}{\epsr}\log \frac{1}{\epsa})$ and $\MTVD(\dR,\tilde\dR) \leq \frac12(\epsr \TVD(\dR) + \epsa)$.
\end{lemma}

The intuition behind the sparsification has been sketched in the technical overview (Section~\ref{sec:tech-overview}).
We divide $[0,\infty]$ into a collection of disjointed intervals, then ``merge'' all the probability masses within one interval together.
We formalized this process as \textsf{SparsifyWrtIntervals} (Algorithm~\ref{alg-wrtintervals}).
The support size of the sparsified ratio is no more than the number of disjoint intervals we choose.
The running time equals the input ratio support size plus the number of intervals,
assuming the ratio is presented by a sorted table
and the intervals is also sorted.

\begin{algorithm}
\caption{$\textsf{SparsifyWrtIntervals}(\dR,\{I_t\}_{1\leq t\leq m})$} \label{alg-wrtintervals}
  \KwIn{a ratio $\dR$ with finite support and a collection of disjoint intervals $I_1,\dots,I_m \subseteq [0,\infty]$}
  % \KwOut{a sparsified ratio $\tilde\dR$}
  Let $(\dR^\dagger,\dR)$ be the canonical pair of $\dR$\;
  Define distributions $\dP^*,\dQ^*$ over $\Omega^* = [0,\infty] - \bigcup_{t}I_t + \{\omega_1,\dots,\omega_m\}$ as
  \[
  \begin{aligned}
    \dP^*(\omega_t) &= \dR^\dagger(I_t) ,&
    \dQ^*(\omega_t) &= \dR(I_t) &
    &~\text{for $1\leq t\leq m$}; \\
    \dP^*(r) &= \dR^\dagger(r),&
    \dQ^*(r) &= \dR(r) &
    &~\text{for $\textstyle r\in [0,\infty] - \bigcup_{t}I_t$}.
  \end{aligned}
  \]
  Equivalently, $\dP^* = \kappa\dR^\dagger$, $\dQ^* = \kappa\dR$ for deterministic Markov kernel $\kappa$
  \[
    \kappa(r) = \begin{cases}
      \omega_i, &\text{ if } r\in I_i \text{ for some } i\in[m]\\
      r, &\text{ otherwise}
    \end{cases}
  \]

  \KwOut{a sparsified ratio $\tilde\dR = (\dP^* \| \dQ^*)$}
\end{algorithm}

As a warm-up, consider how does the sparsification process work on a single interval.
Let $(\dR^\dagger,\dR)$ be the canonical pair of $\dR$,
let $I$ be a sufficiently small interval.
The sparsification of $\dR$ with respect to an interval $I$ works as follows:
\[
  \tilde\dR = (\kappa\dR^\dagger \| \kappa\dR)
  \text{ where $\kappa$ is a deterministic Markov kernel }
  \kappa(r) = \begin{cases}
    \omega^*, &\text{ if } r \in I \\
    r, &\text{ if } r \notin I \\
  \end{cases}
\]
If $\dR^\dagger(I) = \dR(I) = 0$, the sparsification does nothing and $\tilde\dR = \dR$.
Otherwise, 
the sparsified ratio and its alternative are 
\[
  \tilde\dR(r) = \begin{cases}
    \dR(r), &\text{ if }  r\notin I\\
    0, &\text{ if }  r\in I\setminus \{r^*\} \\
    \dR(I), &\text{ if }  r = r^* \\
  \end{cases}
  \qquad
  \tilde\dR^\dagger(r) = \begin{cases}
    \dR^\dagger(r), &\text{ if }  r\notin I\\
    0, &\text{ if }  r\in I\setminus \{r^*\} \\
    \dR^\dagger(I), &\text{ if }  r = r^* \\
  \end{cases}
\]
where $r^* \deq \dfrac{\dR^\dagger(I)}{\dR(I)}$. 
Apparently $r^* \in I$.

To bound $\MTVD(\dR,\tilde\dR)$,
consider the following two distributions $\tilde\dP,\tilde\dQ$
\[
  \tilde\dQ(r) = \begin{cases}
    \dR(r), &\text{ if }  r\notin I\\
    \dR^\dagger(r) / r^*, &\text{ if }  r\in I \\
  \end{cases}
  \qquad
  \tilde\dP(r) = \begin{cases}
    \dR^\dagger(r), &\text{ if }  r\notin I\\
    \dR(r) \cdot r^*, &\text{ if }  r\in I
  \end{cases}
\]
They satisfy $\tilde\dR = (\tilde\dP \| \dR) = (\dR^\dagger \| \tilde\dQ)$.
Therefore, 
% Thus, $\MTVD(\dR,\tilde\dR)$ is less than both $\TVD(\dP,\dR^\dagger)$ and $\TVD(\dQ,\dR)$.
% We will use the tighter one to bound $\MTVD(\dR,\tilde\dR)$.
\[
\begin{aligned}
  \MTVD(\dR,\tilde\dR) \leq
  \TVD(\tilde\dP,\dR^\dagger)
  &= \frac12 \sum_{r\in I} \Bigl| \dR(r) \cdot r^* - \dR^\dagger(r) \Bigr|
  = \frac12 \sum_{r\in I} \dR(r) \cdot | r^* - r |  \,,
  \\
  \MTVD(\dR,\tilde\dR) \leq
  \TVD(\tilde\dQ,\dR)
  &= \frac12 \sum_{r\in I} \Bigl| \dR^\dagger(r) / r^* - \dR(r) \Bigr|
  = \frac12 \sum_{r\in I} \dR^\dagger(r) \cdot | \tfrac1{r^*} - \tfrac1r | \,.
\end{aligned}
\]
We should use the tighter one to bound $\MTVD(\dR,\tilde\dR)$.

Intuitively, $\tilde\dQ$ is obtained by modifying $\dR$ at points $r\in I$;
$\tilde\dP$ is obtained by modifying $\dR^\dagger$ at points $r\in I$.
If the interval $I \subseteq [0,1]$,
then $\dR^\dagger(r) \leq \dR(r)$ for all $r\in I$,
thus modifying $\dR^\dagger$ introduces less error,
in other words, $\TVD(\tilde\dP,\dR^\dagger)$ is a tighter bound of $\MTVD(\dR,\tilde\dR)$.
Symmetrically, if the interval $I \subseteq [1,\infty]$, $\TVD(\tilde\dQ,\dR)$ is a tighter bound.
This intuitive argument can be formalized by:
\[
  \TVD(\tilde{\dQ},\dR)
  = \frac12 \sum_{r\in I} \dR^\dagger(r) \cdot | \tfrac1{r^*} - \tfrac1r | 
  % = \frac12 \sum_{r\in I} \dR^\dagger(r) \cdot r r^* \cdot | {r^*} - r | 
  = \frac12 \sum_{r\in I} \dR(r) \cdot \frac1{r^*} \cdot | {r^*} - r | 
  = \frac1{r^*} \cdot \TVD(\tilde{\dP},\dR^\dagger) \,.
\]

Consider the case $I\subseteq [0,1]$, let $a,b$ be the endpoints of $I$, so $a < b\leq 1$.
Then $\TVD(\tilde{\dP},\dR^\dagger)$ can be upper bounded by one of the two following  arguments:
\begin{itemize}
  \item If $|r^* - r| \leq \epsr (1-r)$ for all $r\in I$,
    then 
    \[
      \TVD(\tilde{\dP},\dR^\dagger)
      = \frac12 \sum_{r\in I} \dR(r) \cdot | r^* - r |
      \leq \frac12 \epsr\sum_{r\in I} \dR(r) \cdot (1-r) \,.
    \]
    By its similarity with $\TVD(\dR) = \sum_{r \leq 1} \dR(r) (1-r)$,
    even if we sparsify the ratio with respect to many disjointed intervals,
    the total error is bounded by $\frac12\epsr \TVD(\dR)$.

    To ensure $|r^* - r| \leq \epsr (1-r)$ for all $r\in I$,
    the two endpoints should satisfy $(b-a) \leq \epsr (1-b)$.
    Note that this cannot be satisfied if $b=1$.

  \item If $|r^* - r| \leq \epsa$ for all $r\in I$,
    then 
    \[
      \TVD(\tilde{\dP},\dR^\dagger)
      = \frac12 \sum_{r\in I} \dR(r) \cdot | r^* - r |
      \leq \frac12 \sum_{r\in I} \dR(r) \cdot \epsa
      \leq \frac\epsa2  \,.
    \]

    To ensure $|r^* - r| \leq \epsa$ for all $r\in I$,
    the two endpoints should satisfy $b-a \leq \epsa$.
\end{itemize}

Based on the above observations, 
we choose the following partition of $[0,\infty]$,
which is a collection of disjoint intervals 
\[
  I_0,I_{1},\dots,I_{m-1},I_m,\{1\},J_m,J_{m-1},\dots,J_1,J_0,
\]
in that order.
% Let $\delta_{a}$ and  
\begin{itemize}
  \item 
    Let $I_t \deq [a_{t},a_{t+1})$ for every $t<m$, where $a_t \deq 1 - (1+\epsr)^{-t}$.
    Thus $a_{t+1} - a_{t} \leq \epsr (1-a_{t+1})$.
  \item
    Let $I_m \deq [a_m,1)$
    and let 
    \begin{align}\label{eq-def-m'}
      m \deq \left \lceil \frac{-\log(\epsa)}{\log(1 + \epsr)} \right \rceil.
    \end{align}
    Thus $1 - a_m \leq \epsa$.
\end{itemize}
Symmetrically, let $J_t \deq (\frac1{a_{t+1}}, \frac1{a_{t}}]$ for $t<m$ and let $J_m \deq (1, \frac1{a_m}]$.

Putting the above intervals into \textsf{SparsifyWrtIntervals} (Algorithm~\ref{alg-wrtintervals}) gives the \textsf{Sparsify} algorithm used in Lemma~\ref{lem:sparsify}.
\begin{equation}
\label{eq:def-sparsify}
   \textsf{Sparsify}(\dR,\epsr,\epsa) =  \textsf{SparsifyWrtIntervals}(\dR, \{ I_0,I_{1},\dots,I_{m-1},I_m,\{1\},J_m,J_{m-1},\dots,J_1,J_0\}).
\end{equation}

\subsection{Proof of the Sparsification Lemma}

This section proves the sparsification lemma (Lemma~\ref{lem:sparsify}).
The sparsified ratio $\tilde \dR = (\kappa\dR^\dagger \| \kappa\dR)$ for $\kappa$ defined in Algorithm~\ref{alg-wrtintervals}, thus $\tilde\dR \leq \dR$.
The support size of the sparsified ratio is no more than the number of intervals, which is $O(\frac{1}{\epsr}\log \frac{1}{\epsa})$.
% is $O(\frac{1}{\epsr}\log \frac{1}{\epsa})$ and they jointly cover $[0,\infty]$, 
%
% the support size of the sparsified ratio is also bounded by $O(\frac{1}{\epsr}\log \frac{1}{\epsa})$.
The running time is proportional to the number of intervals and the support size of the input ratio,
% (assuming the input ratio is represented as a sorted table),
which is $O(\frac{1}{\epsr}\log \frac{1}{\epsa} + |\supp(\dR)|)$.
The rest of the section forces on bounding the error with respect to the minimum total variation distance.

For each $i\in [m]$,
define $r^*_i \in I_i$ as 
\[
  r^*_i 
  \deq \begin{cases}
    \dfrac{\dR^\dagger(I_i)}{\dR(I_i)}, &\text{ if }\dR(I_i) \neq 0\\
    \text{any number in }I_i, &\text{ if }\dR(I_i) = 0\\
  \end{cases}
\]
Note that $\dR(I_i) = 0$ implies $\dR^\dagger(I_i) = 0$.
Symmetrically, define $r^{**}_i \in J_i$ as 
\[
  r^{**}_i \deq
  \begin{cases}
    \dfrac{\dR^\dagger(J_i)}{\dR(J_i)}, &\text{ if } \dR^\dagger(J_i) \neq 0 \\
    \text{any number in }J_i, &\text{ if } \dR^\dagger(J_i) = 0 \\
  \end{cases}
\]
The sparsified ratio $\tilde\dR$ and its alternative $\tilde\dR^\dagger$ can be written as
\[
  \tilde\dR(r) = \begin{cases}
    \dR(I_i), &\text{ if }r=r^*_i \text{ for some }i\in[m]\\
    \dR(J_i), &\text{ if }r=r^{**}_i \text{ for some }i\in[m]\\
    \dR(1), &\text{ if }r=1\\
    0, &\text{ otherwise }\\
  \end{cases}
  \qquad
  \tilde\dR^\dagger(r) = \begin{cases}
    \dR^\dagger(I_i), &\text{ if }r=r^*_i \text{ for some }i\in[m]\\
    \dR^\dagger(J_i), &\text{ if }r=r^{**}_i \text{ for some }i\in[m]\\
    \dR^\dagger(1), &\text{ if }r=1\\
    0, &\text{ otherwise }\\
  \end{cases}
\]

To bound $\MTVD(\dR,\tilde\dR)$, consider the following two distributions,
as inspired by the discussion of sparsification w.r.t.\ one interval,
\begin{equation*}
  \tilde\dQ(r) = \begin{cases}
    \dR(r), &\text{ if }  r\leq 1\\
    \dR^\dagger(r) / r^{**}_i, &\text{ if } r\in J_i \\
  \end{cases}
  \qquad
  \tilde\dP(r) = \begin{cases}
    \dR^\dagger(r), &\text{ if }  r\geq 1\\
    \dR(r) \cdot r^*, &\text{ if }  r\in I_i 
  \end{cases}
\end{equation*}

\begin{claim}
$\tilde\dP,\tilde\dQ$ are well-defined distributions and $\tilde\dR = (\tilde\dP \| \tilde\dQ)$.
\end{claim}
\begin{proof}
  For each $i\in[m]$,
  $\tilde\dQ(J_i) = \dR^\dagger(J_i) / r^{**}_i = \dR(J_i)$.
  Thus
  \[
    \tilde\dQ([0,\infty])
    = \tilde\dQ([0,1]) + \sum_{i\in[m]} \tilde\dQ(J_i)
    = \dR([0,1]) + \sum_{i\in[m]} \dR(J_i)
    = \dR([0,\infty])
    = 1 \,.
  \]
  So $\tilde\dQ$ is a distribution over $[0,\infty]$.
  Symmetrically, $\tilde\dP(I_i) = \dR(I_i) \cdot r^*_i = \dR^\dagger(I_i)$,
  and $\tilde\dP$ is a distribution.

  Say random variable $X\sim \tilde\dQ$,
  \begin{itemize}
    \item $X\in I_i \iff \frac{\tilde\dP(X)}{\tilde\dQ(X)} = r^*_i$ and  $\Pr[X\in I_i] = \tilde\dQ(I_i) = \dR(I_i)$.
    \item $X\in J_i \iff \frac{\tilde\dP(X)}{\tilde\dQ(X)} = r^{**}_i$ and $\Pr[X\in J_i] = \tilde\dQ(J_i) = \dR(J_i)$.
    \item $X= 1 \iff \frac{\tilde\dP(X)}{\tilde\dQ(X)} = 1$ and $\Pr[X = 1] = \tilde\dQ(1) = \dR(1)$.
  \end{itemize} 
  Thus $\tilde\dR = (\tilde\dP \| \tilde\dQ)$.
\end{proof}

Then $\MTVD(\tilde\dR,\dR)$ is upper bounded by
$\MTVD(\tilde\dR,\dR) \leq \max(\TVD(\tilde\dP,\dR^\dagger), \TVD(\tilde\dQ,\dR))$.
Consider the two terms on the right hand side separately.
\[
  \TVD(\tilde\dP,\dR^\dagger)
  = \frac12\sum_{r<1} \bigl|\tilde\dP(r) - \dR^\dagger(r)\bigr|
  = \frac12\sum_{i=1}^m \sum_{r\in I_i} \bigl|\tilde\dP(r) - \dR^\dagger(r)\bigr|
  = \frac12\sum_{i=1}^m \sum_{r\in I_i} \dR(r) \cdot \bigl|\tilde r^*_i - r\bigr|
\]
By how we choose the intervals, for each $i\in[m]$ and any $\tilde r^*_i,r\in I_i$,
we have
\[
  \bigl|r^*_i - r\bigr|
  \leq \begin{cases}
    \epsr (1-r), &\text{ if } i<m \\
    \epsa, &\text{ if } i = m \\
  \end{cases}
\]
Therefore,
\[
\begin{aligned}
  \TVD(\tilde\dP,\dR^\dagger)
  &\leq \frac12\sum_{i=1}^{m-1} \sum_{r\in I_i} \dR(r) \cdot \epsr (1-r)
  + \frac12 \sum_{r\in I_m} \dR(r) \cdot \epsa \\
  &\leq \frac12 \sum_{r<1} \dR(r) \cdot \epsr (1-r) + \frac12 \epsa \\
  &= \frac12 \epsr \TVD(\dR) + \frac12 \epsa \,.
\end{aligned}
\]

Symmetrically, for each $i\in[m]$ and any $r^{**}_i,r\in J_i$,
we have
\[
  \bigl|\frac1{r^{**}_i} - \frac1r \bigr|
  \leq \begin{cases}
    \epsr (1-\frac1r), &\text{ if } i<m \\
    \epsa, &\text{ if } i = m \\
  \end{cases}
\]
due to the choice of intervals.
Then $\TVD(\tilde\dQ,\dR)$ can be upper bounded by
\[
\begin{aligned}
  \TVD(\tilde\dQ,\dR)
  % &= \frac12 \sum_{r>1} \bigl| \tilde\dQ(r) - \dR(r) \bigr| 
  % = \frac12 \sum_{i\in[m]} \sum_{r\in J_i} \bigl| \tilde\dQ(r) - \dR(r) \bigr| \\
  % &= \frac12 \sum_{i=1}^m \sum_{r\in J_i} \dR^\dagger(r) \cdot \bigl| \frac1{r^{**}_i} - \frac1r \bigr| \\
  % &\leq \frac12 \sum_{i=1}^{m-1} \sum_{r\in J_i} \dR^\dagger(r) \cdot \epsr \Bigl( 1 - \frac1r \Bigr) 
  % + \frac12 \sum_{r\in J_m} \dR^\dagger(r) \cdot \epsa  \\
  &\leq \frac12 \sum_{r > 1} \dR^\dagger(r) \cdot \epsr \Bigl( 1 - \frac1r \Bigr) 
  + \frac12 \cdot \epsa  
  = \frac12 \epsr \TVD(\dR) + \frac12 \epsa \,.
\end{aligned}
\]
The last equality symbol relies on 
\[
  \TVD(\dR)
  = \TVD(\dR^\dagger,\dR)
  = \sum_{r \text{ s.t.\ } \dR^\dagger(r) > \dR(r)} (\dR^\dagger(r) - \dR(r))
  = \sum_{r>1} \dR^\dagger(r) \Bigl( 1 - \frac1r \Bigr) \,.
\]

As the conclusion, 
$\MTVD(\tilde\dR,\dR) \leq \max(\TVD(\tilde\dP,\dR^\dagger), \TVD(\tilde\dQ,\dR)) \leq \frac12(\epsr \TVD(\dR) + \epsa)$.

\section{Estimate TV Distance between Product Distributions}
\label{sec:product}

The two product distributions of interest are 
$\dP=\dP_1 \dP_2 \ldots \dP_n$ and $\dQ=\dQ_1 \dQ_2 \ldots \dQ_n$,
where $\dP_1,\allowbreak\dP_2,\ldots,\dP_n,\allowbreak\dQ_1,\dQ_2,\ldots,\dQ_n$ are $2n$ distributions over sample space $[q]$.
We present a deterministic algorithm that estimates the TV distance between $\dP$ and $\dQ$.
The intuition has be explained in the technical overview (Section~\ref{sec:overview}).

\begin{theorem}
\label{thm:product}
\label{thm-main-1}
  There exists a deterministic algorithm (Algorithm~\ref{alg:product}).
  Given two product distributions $\dP,\dQ$ over $[q]^n$ and $\varepsilon >0$, 
  it outputs an estimation $\widehat{\Delta}$ satisfying $(1-\varepsilon)\TVD(\dP, \dQ) \leq \widehat{\Delta} \leq \TVD(\dP, \dQ)$,
  in time $O\bigl(\frac{qn^2}{\varepsilon} \log q \log \frac{n}{\varepsilon \TVD(P,Q)}\bigr)$.
\end{theorem}

\begin{algorithm}
\caption{A Deterministic FPTAS for the TV distance between product distributions} \label{alg:product}
  \KwIn{$2n$ distributions $\dP_1,\dP_2,\ldots,\dP_n, \dQ_1,\dQ_2,\ldots,\dQ_n$, and an error bound $\varepsilon > 0$;}
  % \KwOut{an $\varepsilon$-approximation to $\TVD(\dP,\dQ)$;}
  %Initialise the set $\Omega_R = \{1\}$ and the function $f_R$ as $f_R(1) = 1$\\
  Let $\LBTV \gets \max_{1 \leq i \leq n}\TVD(\dP_i,\dQ_i)$\\
  Let $\dR'_{1:1} \gets \dR_1$, where $\dR_1 \deq (\dP_1\|\dQ_1)$\\
  \For{$k$ from $1$ to $n-1$}{
    $\tilde\dR_{1:k} \gets \textsf{Sparsify}(\dR'_{1:k},\frac{\varepsilon}{2n},\frac{\varepsilon}{2n} \LBTV)$\\ %, where $\delta = \frac{\varepsilon}{2n}$\\
    $\dR'_{1:k+1} \gets \tilde\dR_{1:k} \indpprod \dR_{k+1}$, where $\dR_{k+1} \deq (\dP_{k+1}||\dQ_{k+1}) $\\
  }
  \Return $\widehat{\Delta}= \TVD(\dR'_{1:n}) $\\
\end{algorithm}

The algorithm starts by computing a relatively tight lower bound of the $\TVD(\dP,\dQ)$.
\begin{claim}
  $\frac{\TVD(\dP,\dQ)}{n} \leq  \LBTV \leq \TVD(\dP,\dQ)$. 	
\end{claim}
\begin{proof}
  $\TVD(\dP_i,\dQ_i) \leq \TVD(\dP,\dQ)$ for every $i\in[n]$, thus $\LBTV \leq \TVD(\dP,\dQ)$. 

  Let $(X_1,Y_1),\dots,(X_n,Y_n)$ be independent and $(X_i,Y_i)$ is the optimal coupling of $(\dP_i,\dQ_i)$.
  % Let $X=(X_1,\dots,X_n) \sim\dP, Y=(Y_1,\dots,Y_n) \sim\dQ$.
  Hence, 
  \[
    \TVD(\dP,\dQ) \leq  \Pr[(X_1,\dots,X_n) \neq (Y_1,\dots,Y_n) ] \leq \sum_{i=1}^n\Pr[X_i \neq Y_i] =\sum_{i=1}^n \TVD(\dP_i,\dQ_i) \leq n \LBTV \,.
    \qedhere
  \]
\end{proof}

The algorithm then iteratively computes $\dR_{1:k}'$ for $k = 1$ up to $n$.
$\dR_{1:k}'$ is an approximation of the ``actual'' ratio 
$\dR_{1:k} \deq (\dP_1\dP_2\dots\dP_k \| \dQ_1\dQ_2\dots\dQ_k)$.
% between $\dP_1\dP_2\dots\dP_k$ and $\dQ_1\dQ_2\dots\dQ_k$.
% Denoted the actual ratio by
% \[
%   \dR_{1:k} \deq (\dP_1\dP_2\dots\dP_k \| \dQ_1\dQ_2\dots\dQ_k)
% \]
\begin{claim}
  For each $k \in [n]$,
  it holds that $\dR_{1:k}' \leq \dR_{1:k}$ 
  and $\MTVD(\dR_{1:k}', \dR_{1:k}) \leq \frac{k-1}{2n}\varepsilon \TVD(\dP,\dQ)$. 
\end{claim}

\begin{proof}
  It is proved by induction on $k$.
  The base case is satisfied as $\dR_{1:1}' = (\dP_1\|\dQ_1) = \dR_{1:1}$.

  For larger $k$, the inductive assumption is
  $\dR_{1:k-1}' \leq \dR_{1:k-1}$ and $\MTVD(\dR_{1:k-1}', \dR_{1:k-1}) \leq \frac{k-2}{2n}\varepsilon \TVD(\dP,\dQ)$.
  By the data processing inequality for TV distance,
  % $\dR_{1:k-1}' \leq \dR_{1:k-1}$ also implies 
  $\TVD(\dR_{1:k-1}') \leq \TVD(\dR_{1:k-1}) \leq \TVD(\dP,\dQ)$.
  By the nature of the sparsification process (Lemma~\ref{lem:sparsify}), 
  the sparsified ratio $\tilde\dR_{1:k-1} \leq \dR_{1:k-1}'$
  and
  \begin{align}\label{eq-bound-1}
    \MTVD(\tilde\dR_{1:k-1}, \dR_{1:k-1}')
    \leq \frac12\bigl( \frac{\varepsilon}{2n} \TVD(\dR_{1:k-1}') + \frac{\varepsilon}{2n} \LBTV \bigr)
    \leq \frac{\varepsilon}{2n} \TVD(\dP,\dQ) \,.
  \end{align}
  Then the inductive proof is concluded by
  % by Proposition~\ref{prop:product-order}
  % and Lemma~\ref{lem:product-MTVD}
  \[
    \dR_{1:k}'
    = \tilde\dR_{1:k-1} \indpprod \dR_k
    \underset{\text{Proposition~\ref{prop:product-order}}}\leq \dR_{1:k-1} \indpprod \dR_k
    = \dR_{1:k} \,,
  \]
  \[
  \begin{aligned}[b]
    &\MTVD(\dR_{1:k}', \dR_{1:k})
    \underset{\text{Lemma~\ref{lem:product-MTVD}}}\leq \MTVD(\tilde\dR_{1:k-1}, \dR_{1:k-1}) \\
    \leq \,&\MTVD(\tilde\dR_{1:k-1}, \dR_{1:k-1}') + \MTVD(\dR_{1:k-1}, \dR_{1:k-1}')
    \underset{\text{I.H. and }\eqref{eq-bound-1}}\leq \frac{1}{2n} \varepsilon \TVD(\dP,\dQ) + \frac{k-2}{2n}\varepsilon \TVD(\dP,\dQ) \,.
  \end{aligned}\qedhere
  \]
\end{proof}

% \subsection{The Analysis (Proof of \Cref{thm-main-1})}

In the end, the algorithm outputs $\widehat{\Delta}= \TVD(\dR'_{1:n})$ as an estimation of the TV distance.
Since $\MTVD(\dR'_{1:n}, \dR_{1:n}) \leq \frac12\varepsilon \TVD(\dP,\dQ)$,
we have
\[
  |\widehat{\Delta} - \TVD(\dP,\dQ)|
  = |\TVD(\dR'_{1:n}) - \TVD(\dR_{1:n})|
  \leq  \varepsilon \TVD(\dP,\dQ) \,.
\]
In the meanwhile, $\dR'_{1:n} \leq \dR_{1:n}$ implies
\[
  \widehat{\Delta} = \TVD(\dR'_{1:n}) \leq \TVD(\dR_{1:n}) = \TVD(\dP,\dQ) \,.
\]

The running time of our algorithm is dominated by the main loop.
% In each iteration,
The support size of $\tilde\dR_{1:k}$ is bounded by 
$O\bigl(\frac{n}{\varepsilon} \ln \frac{n}{\varepsilon \LBTV}\bigr) \leq O \bigl(\frac{n}{\varepsilon} \ln \frac{n}{\varepsilon \TVD(P,Q)}\bigr)$,
by Lemma~\ref{lem:sparsify}.
The support size of $\dR_{1:k+1}'$ is at most $q$ times larger.
The time complexity of computing $\dR_{1:k+1}'$ is $O\bigl(\frac{qn}{\varepsilon} \log q \log \frac{n}{\varepsilon \TVD(P,Q)}\bigr)$,
which is the complexity of merging $q$ sorted list, each list of length $O\bigl(\frac{n}{\varepsilon} \log \frac{n}{\varepsilon \TVD(P,Q)}\bigr)$.
The time complexity of computing $\tilde\dR_{1:k+1}$ from $\dR_{1:k+1}'$ is smaller.
Thus the time complexity of each iteration is dominated by computing $\dR_{1:k+1}'$,
and the total time complexity is $O\bigl(\frac{qn^2}{\varepsilon} \log q \log \frac{n}{\varepsilon \TVD(P,Q)}\bigr)$.

% Algorithm \ref{alg:product} has $n$ iterations in total.
% By~Algorithm \ref{alg:product} and~\eqref{eq-def-m'}, the support size of each sparsified distribution is at most
% \begin{align*}
% 	2m + 3 = O\tp{\frac{ \ln\frac{1}{\delta \LBTV } }{\ln(1+\delta)}} = O\tp{\frac{1}{\delta} \cdot \ln \frac{1}{\delta \LBTV}} = O\tp{\frac{n}{\varepsilon} \ln \frac{n}{\varepsilon \LBTV}} =  O\tp{\frac{n}{\varepsilon} \ln \frac{n}{\varepsilon \TVD(P,Q) }}, 
% \end{align*}
% where the last equation follows from \Cref{prop-lb}. By~\Cref{thm:spar}, the running time of each iteration can be bounded by $O\tp{\frac{qn}{\varepsilon} \log \frac{qn}{\varepsilon \TVD(P,Q) }}$.
% The total running time is $O\tp{\frac{qn^2}{\varepsilon} \log \frac{qn}{\varepsilon \TVD(P,Q) }}$.

% \tianren{New complexity  $O\bigl(\frac{qn^2}{\varepsilon} {\color{red}\log q} \log \frac{n}{\varepsilon \TVD(P,Q)}\bigr)$ 
% Old complexity $O\tp{\frac{qn^2}{\varepsilon} \log \frac{{\color{red}q}n}{\varepsilon \TVD(P,Q) }}$}

% \input{markov}
\section{Estimate TV Distance between Markov Chains}\label{sec-markov-setting-new}\label{sec:markov}

% \liqiang{In order to approximate the TV-distance between two Markov Chain trajectories (as formally defined in the next paragraph), one can still follow our methodology: taking the $n$ random variables into account one-by-one and use \emph{Sparsification} to give an approximation of the current TV-distance, at the same time keeping the support of distributions relatively small (recap Algorithm~\ref{alg:product}). The major obstacle is that, in the case of Markov chain, one cannot easily take a new random variable into the ratio by simply doing independent product, as shown in line (3) of Algorithm~\ref{alg:product}. Thus, as put in this section, instead of \emph{ratio}, we focus on \emph{conditional ratio}, of which the probability distributions are both "conditional distribution", and one can observe a linear relation as the counterpart of line (3) in Algorithm~\ref{alg:product} (we call it \emph{Concatenate} in this section).}

A Markov chain $\dP$ over $[q]^n$ can be represented by 
its initial distribution $\dP_1$ and $n-1$ Markov kernels $\kappa_2,\dots,\kappa_n$
such that 
$\dP(x_1,x_2,\dots,x_n) = \dP_1(x_1) \kappa_2(x_2|x_1) \dots \kappa_n(x_n|x_{n-1})$.
In the Markov chain setting, we use the conventional notation $\dP\MK{k}{k-1}$ to denote the Markov kernel $\kappa_k$.
As the notation suggested,
if $(X_1,\dots,X_n) \sim \dP$, 
\[
  \dP\MK{k}{k-1}(x_{k} | x_{k-1}) = \Pr[X_k = x_{k} \mid X_{k-1} = x_{k-1}] 
\]
for all $x_{k-1},x_k \in [q]$ such that $\Pr[X_{k-1} = x_{k-1}] > 0$.

Similarly, another Markov chain $\dQ$ is represented by its initial distribution $\dQ_1$ and Markov kernels $\dQ\MK21, \dots, \dQ\MK{n}{n - 1}$.
% The goal is to approximate the TV distance between $\dP$ and $\dQ$.
We give a deterministic FPTAS for the TV distance between $\dP$ and $\dQ$.

\begin{theorem}\label{thm-main-2}\label{thm:markov}
  There is a deterministic algorithm (Algorithm~\ref{alg:markov}) such that
  given two Markov chains $\dP,\dQ$ over $[q]^n$ and an error bound $\varepsilon > 0$, 
  it outputs $\widehat{\Delta}$ satisfying $(1-\varepsilon)\TVD(\dP, \dQ) \leq \widehat{\Delta} \leq \TVD(\dP, \dQ)$ 
  in time $O(\frac{n^2q^2}{\varepsilon} \log q  \log (\frac{n}{\varepsilon \TVD(\dP,\dQ)}))$.
\end{theorem}

To present our algorithm,
we introduce a few more notations, in particular, the \emph{conditional ratio between Markov kernels},
which plays the central role in our algorithm.

The conventional notation $\dP_{k:n|k-1}$ denotes a Markov kernel from $[q]$ to $[q]^{n-k+1}$,
is defined as
\[
  \dP\MK{k:n}{k-1}(x_k,x_{k+1},\dots,x_n|x_{k-1})
  \deq
  \dP\MK{n}{n-1}(x_n|x_{n-1})
  \dP\MK{n-1}{n-2}(x_{n-1}|x_{n-2})
  \dots
  \dP\MK{k}{k-1}(x_k|x_{k-1}) \,.
\]
As the notation suggested,
if $(X_1,\dots,X_n) \sim \dP$, 
\[
  \dP\MK{k:n}{k-1}(x_k,x_{k+1},\dots,x_n|x_{k-1}) = \Pr[\, (X_k,\dots,X_n) = (x_{k},\dots,x_n) \mid X_{k-1} = x_{k-1}] 
\]
for all $x_{k-1},x_k,\dots,x_n \in [q]$ such that $\Pr[X_{k-1} = x_{k-1}] > 0$.
We use $\dP\MK[x]{k}{k-1}$ and $\dP\MK[x]{k:n}{k-1}$ to denote the derived distributions 
$\dP\MK{k}{k-1}( \dots | x )$ and $\dP\MK{k:n}{k-1}( \dots | x )$ respectively.
That is
\[
  (\dP\MK[x_{k-1}]{k}{k-1})(x_k) \deq \dP\MK{k}{k-1}( x_k | x_{k-1} ) \,,\qquad
  (\dP\MK[x_{k-1}]{k:n}{k-1})(x_{k:n}) \deq \dP\MK{k:n}{k-1}( x_{k:n} | x_{k-1}) \,.
\]
They can be viewed as the conditional distributions of $X_k$ and $X_{k:n}$.

Then we can consider the ratio between $\dP\MK[x]{k:n}{k-1}$ and $\dQ\MK[x]{k:n}{k-1}$, defined as
\[
  \dR\MK[x]{k:n}{k-1} \deq
  \bigl( \dP\MK[x]{k:n}{k-1} \bigm\| \dQ\MK[x]{k:n}{k-1} \bigr) \,.
\]
This implicitly defines a Markov kernel $\dR\MK{k:n}{k-1}$ from $[q]$ to $[0,\infty)$,
as formalized below.

\begin{definition}[Conditional Ratio between Markov kernels]
  For two Markov kernels $\dP\MK YX,\dQ\MK YX$ from $\Omega_X$ to $\Omega_Y$,
  the \emph{conditional ratio} between them, denoted as $(\dP\MK YX \| \dQ\MK YX)$, is a Markov kernel from $\Omega_X$ to $[0,\infty)$,
  such that, if letting $\dR\MK YX \deq (\dP\MK YX \| \dQ\MK YX)$, for every $x\in\Omega_X$
  \[
    \dR\MK[x]YX = \bigl(\dP\MK[x]YX \bigm\| \dQ\MK[x]YX\bigr) \,.
  \]
  A \emph{valid conditional ratio} over $\Omega_X$ is Markov kernel $\kappa$ from $\Omega_X$ to $[0,\infty)$,
  such that $\kappa(\cdot|x)$ is a valid ratio for every $x\in\Omega_X$,
\end{definition}

Conditional ratio plays the central role in our algorithm, as presented in Algorithm~\ref{alg:markov}.
The core of the algorithm is an iterative process computing an approximation of the conditional ratio $\dR\MK{k:n}{k-1}$ for each $k$.
In each iteration step,
to compute an approximation of $\dR\MK[x]{k:n}{k-1}$,
the algorithm need an approximation of $\dR\MK{k+1:n}{k}$ from last iteration,
and distributions $\dP\MK[x]{k}{k-1}, \dQ\MK[x]{k}{k-1}$ from the inputs.
% We name the process of computing (approximate) 
We call this subroutine \textsf{Concatenate}, and it is defined and analyzed in Section~\ref{sec:concat}.

\begin{algorithm}
  \caption{A Deterministic FPTAS for the TV distance between two Markov chains} \label{alg-markov-new}\label{alg:markov}
  \KwIn{two Markov chains $\dP = \dP_1\dP_{2|1}\dots\dP_{n|n-1}$, $\dQ = \dQ_1\dQ_{2|1}\dots\dQ_{n|n-1}$, \newline
        an error bound $\varepsilon > 0$}
  Compute $\LBTV$ such that $\frac{\TVD(\dP,\dQ)}{2n} \leq \LBTV \leq \TVD(\dP,\dQ)$ (Section~\ref{sec:markov-lb})\\
  Compute conditional ratio $\dR'\MK{n:n}{n-1} \gets \bigl( \dP\MK{n}{n-1} \bigm\| \dQ\MK{n}{n-1} \bigr)$ \newline
  That is, $\dR'\MK[x]{n:n}{n-1} \gets \bigl( \dP\MK[x]{n}{n-1} \bigm\| \dQ\MK[x]{n}{n-1} \bigr)$ for each $x\in[q]$ \\
  \For{$k$ from $n-1$ to $1$}{
    Compute conditional ratio $\tilde\dR\MK{k+1:n}{k}$ as follows: \For{$x$ in $[q]$}{
      $\tilde\dR\MK[x]{k+1:n}{k} \gets \textsf{Sparsify}(\dR'\MK[x]{k+1:n}{k},~ \frac{\varepsilon}{4n},~ \frac{\varepsilon}{2n} \LBTV)$
    }
    \eIf{$k>1$}{
      Compute conditional ratio $\dR'\MK{k:n}{k-1}$ as follows: \For{$x$ in $[q]$}{
        $\dR'\MK[x]{k:n}{k-1} \gets \textsf{Concatenate}(\dP\MK[x]{k}{k-1},~ \dQ\MK[x]{k}{k-1},~ \tilde\dR\MK{k+1:n}{k})$  
      }
    }{
      Compute ratio $\dR'_{1:n} \gets \textsf{Concatenate}(\dP_1,~ \dQ_1,~ \tilde\dR\MK{2:n}{1})$
    }
  }
  \KwOut{$\widehat{\Delta}= \TVD(\dR'_{1:n})$}
\end{algorithm}

\subsection{The Concatenation of Conditional Ratios}
\label{sec:concat}

Given $\dP_X,\dP\MK YX$ and $\dQ_X,\dQ\MK YX$, 
to compute the ratio $\dR_{XY} \deq ( \dP_X\dP\MK YX \| \dQ_X\dQ\MK YX )$,
the na\"ive solution is to first compute the two joint distributions $\dP_X\dP\MK YX, \dQ_X\dQ\MK YX$.
This section shows that there is an alternative approach:
first compute $\dR\MK YX \deq (\dP\MK YX \| \dQ\MK YX)$,
then
\begin{center}
  $\dR_{XY}$ can be computed from $\dP_X$, $\dQ_X$ and $\dR\MK YX$.
\end{center}

% To define random variable $R$ such that $R \sim ( \dP_X\dP\MK YX \| \dQ_X\dQ\MK YX )$,
Let $(X,Y) \sim \dQ_X\dQ\MK YX$ and set
\[
  R = \frac{(\dP_X\dP\MK YX)(X,Y)}{(\dQ_X\dQ\MK YX)(X,Y)} 
  = \frac{\dP_X(X)}{\dQ_X(X)}  \frac{\dP\MK YX(Y|X)}{\dQ\MK YX(Y|X)} \,.
\]
then $R \sim \dR_{XY} = ( \dP_X\dP\MK YX \| \dQ_X\dQ\MK YX )$.
What is the distribution of $R$ conditioning on $X=x$?
Let $\dR_x$ denote this conditional distribution.
Then $\dR_x$ is the distribution of 
\begin{equation}
\label{eq:concat-x}
  \frac{\dP_X(x)}{\dQ_X(x)}
  \frac{\dP\MK YX(Y_x|x)}{\dQ\MK YX(Y_x|x)}
\end{equation}
where $Y_x$ is sampled from the distribution of $Y$ conditioning on $X = x$.
In other words, $Y_x \sim \dQ\MK[x]YX$.
In such case, the distribution of the second fraction in \eqref{eq:concat-x} is $\dR\MK[x]YX$. 
Therefore, $\dR_x$ can be simply computed from $\dR\MK[x]YX$,
and the wanted ratio $\dR_{XY}$ is a weighted average of $\dR_x$.

If we let $\Dist_{X\sim\dQ}[f(X)]$ denote the distribution of $f(X)$ when $X\sim\dQ$, 
the above analysis can be written as one formula.
% \[
% \begin{aligned}
%   \dR_{XY} 
%   &= ( \dP_X\dP\MK YX \| \dQ_X\dQ\MK YX ) \\
%   &= \Dist_{(X,Y) \sim \dQ_X\dQ\MK YX} \Bigl[ \frac{\dP_X(X)}{\dQ_X(X)}  \frac{\dP\MK YX(Y|X)}{\dQ\MK YX(Y|X)} \Bigr] \\
%   % &= \Dist_{(X,Y) \sim \dQ_X\dQ\MK YX} \Bigl[ \frac{\dP_X(X)}{\dQ_X(X)}  \frac{\dP\MK YX(Y|X)}{\dQ\MK YX(Y|X)} \Bigr] \\
%   &= \sum_{x} \dQ(x)  \Dist_{Y \sim \dQ\MK[x] YX} \Bigl[ \frac{\dP_X(x)}{\dQ_X(x)}  \frac{\dP\MK YX(Y|x)}{\dQ\MK YX(Y|x)} \Bigr] \\
%   &= \sum_{x} \dQ(x)  \Dist_{R \sim \dR\MK[x] YX} \Bigl[ \frac{\dP_X(x)}{\dQ_X(x)} \cdot R \Bigr] \\
% \end{aligned}
% \]
\begin{multline}
\label{eq:concat}
  ( \dP_X\dP\MK YX \| \dQ_X\dQ\MK YX ) 
  = \Dist_{(X,Y) \sim \dQ_X\dQ\MK YX} \Bigl[ \frac{\dP_X(X)}{\dQ_X(X)}  \frac{\dP\MK YX(Y|X)}{\dQ\MK YX(Y|X)} \Bigr] \\
  = \sum_{x} \dQ_X(x)  \Dist_{Y \sim \dQ\MK[x] YX} \Bigl[ \frac{\dP_X(x)}{\dQ_X(x)}  \frac{\dP\MK YX(Y|x)}{\dQ\MK YX(Y|x)} \Bigr] 
  = \sum_{x} \dQ_X(x)  \Dist_{R \sim \dR\MK[x] YX} \Bigl[ \frac{\dP_X(x)}{\dQ_X(x)} \cdot R \Bigr] 
\end{multline}

\begin{lemma}[Concatenation]\label{lem:concat}
  There is a deterministic algorithm \textsf{Concatenate} inspired by \eqref{eq:concat}.
  Given distributions $\dP_X,\dQ_X$ over $\Omega_X$ and a conditional ratio $\dR\MK YX$ over $\Omega_X$,
  the algorithm outputs $\dR_{XY} = \textnormal{\textsf{Concatenate}} (\dP_X,\dQ_X,\dR\MK YX)$
  such that
  \[
    \dR_{XY} = ( \dP_X\dP\MK YX \| \dQ_X\dQ\MK YX )
  \]
  for any sample space $\Omega_Y$ and
  any Markov kernels $\Omega_X$ to $\Omega_Y$
  satisfying $(\dP\MK YX \| \dQ\MK YX) = \dR\MK YX$.

  If the support size of $\dR\MK[x] YX$ is less than $N$ for every $x \in \Omega_X$,
  the support size of $\dR_{XY}$ is no more than $N |\Omega_X|$,
  the time complexity of \textsf{Concatenate} is bounded by $N |\Omega_X| \log (|\Omega_X|)$.
\end{lemma}

\begin{proof}
  The correctness is ensured by \eqref{eq:concat}.

  The complexity bottleneck is to merge $|\Omega_X|$ sorted lists, 
  each list has up to $N$ elements.
\end{proof}

The concatenation process has a few nice properties.
In particular, concatenation preserves the order relation (Definition~\ref{def:order})
and does not increase minimum total variation distance.

\begin{lemma}\label{lem:concat-order}
  For any conditional ratios $\dR\MK YX, \dR'\MK YX$ over $\Omega_X$ such that $\dR\MK[x]YX \geq \dR'\MK[x]YX$ for every $x\in\Omega_X$,
  for any distributions $\dP_X,\dQ_X$ over $\Omega_X$,
  it holds that
  \[
    \textnormal{\textsf{Concatenate}} (\dP_X,\dQ_X,\dR\MK YX) \geq \textnormal{\textsf{Concatenate}} (\dP_X,\dQ_X,\dR'\MK YX) \,.
  \]
\end{lemma}

\begin{proof}
  Since $\dR\MK[x]YX \geq \dR'\MK[x]YX$ for every $x\in\Omega_X$,
  there exist a sample space $\Omega^{(x)}_Y$,
  distributions $\dP^{(x)},\dQ^{(x)}$ over $\Omega^{(x)}_Y$,
  and a Markov kernel $\kappa_x$ from $\Omega^{(x)}_Y$ to $\Omega^{(x)}_Y$
  such that 
  \[
    \dR\MK[x]YX = (\dP^{(x)} \| \dQ^{(x)}) \,,\qquad
    \dR'\MK[x]YX = (\kappa_x\dP^{(x)} \| \kappa_x\dQ^{(x)}) \,.
  \]
  Define sample space $\Omega_Y = \bigcup_{x} \Omega^{(x)}_Y$,
  Markov kernels $\dP\MK YX,\dP'\MK YX,\dQ\MK YX,\dQ'\MK YX$ from $\Omega_X$ to $\Omega_Y$
  % and Markov kernels $\kappa_x$ from $\Omega_Y$ to $\Omega_Y$ for each $x\in\Omega_X$
  as
  \[
    \dP\MK[x]YX \deq \dP^{(x)} ,\quad
    \dQ\MK[x]YX \deq \dQ^{(x)} ,\quad
    \dP'\MK[x]YX \deq \kappa_x\dP^{(x)} ,\quad
    \dQ'\MK[x]YX \deq \kappa_x\dQ^{(x)} .
  \]
  Then $\dR\MK YX = (\dP\MK YX \| \dQ\MK YX)$ and $\dR'\MK YX = (\dP'\MK YX \| \dQ'\MK YX)$.
  So the two sides of the wanted inequality equal
  $(\dP_X\dP\MK YX  \| \dQ_X\dQ\MK YX )$
  and $(\dP_X\dP'\MK YX \| \dQ_X\dQ'\MK YX)$ respectively.
  % so
  % \[
  % \begin{aligned}
  %   \bigl(\dP_X\dP\MK YX  \bigm\| \dQ_X\dQ\MK YX \bigr) &= \textnormal{\textsf{Concatenate}} (\dP_X,\dQ_X,\dR\MK YX) \,, \\
  %   \bigl(\dP_X\dP'\MK YX \bigm\| \dQ_X\dQ'\MK YX\bigr) &= \textnormal{\textsf{Concatenate}} (\dP_X,\dQ_X,\dR'\MK YX) \,.
  % \end{aligned}
  % \]

  Define a Markov kernel $\kappa$ from $\Omega_X\times\Omega_Y$ to itself as
  \[
    \kappa(x',y' | x,y)
    \deq \1[x'=x] \cdot \kappa_x(y'|y) .
  \]
  Intuitively, the distribution of $\kappa$ conditioning on $(x,y)$
  is to set $X' = x$ and sample $Y'$ according to $\kappa_x(\cdot|y)$.
  It is easy to verify that
  $\kappa (\dP_X\dP\MK YX) = \dP_X\dP'\MK YX$
  and 
  $\kappa (\dQ_X\dQ\MK YX) = \dQ_X\dQ'\MK YX$.
  Therefore,
  \begin{multline*}
    \textnormal{\textsf{Concatenate}} (\dP_X,\dQ_X,\dR\MK YX) = \bigl(\dP_X\dP\MK YX  \bigm\| \dQ_X\dQ\MK YX \bigr)  \\
    \geq \bigl(\dP_X\dP'\MK YX \bigm\| \dQ_X\dQ'\MK YX\bigr) = \textnormal{\textsf{Concatenate}} (\dP_X,\dQ_X,\dR'\MK YX) \,.
    \qedhere
  \end{multline*}
\end{proof}

\subsection{A Lower Bound for TV Distance}
\label{sec:markov-lb}

Similar to the case of product distributions (Section~\ref{sec:product}),
our algorithm approximating the TV distance between Markov chains 
needs a relative tight lower bound to start with.
% Similar to the algorithm for the TV distance between two product distributions, we need a lower bound for the TV distance between two Markov chains.
\begin{lemma}\label{lem-lb-mar}
  There exists an deterministic algorithm such that 
  given two Markov chains $\dP$ and $\dQ$ over $[q]^n$, 
  it returns $\LBTV$ satisfying $\frac{\TVD(\dP,\dQ)}{2n} \leq \LBTV \leq \TVD(\dP,\dQ)$ in time $O(q^2n)$.
\end{lemma}
\begin{proof}
  Define $\dD_0 \deq \dP$. 
  For each $k\in[n]$, define distribution $\dD_k \deq \dQ_1\cdots \dQ_{k|k-1} \dP_{k+1|k} \cdots \dP_{n|n-1}$.
  % which is the distribution of $X_1,\dots,X_n$ such that $(X_1,\ldots,X_{i}) \sim \dQ_{1:i}$ and $(X_{i+1}\ldots X_n)$ is obtained by $X_{i}$ and $\dP_{j|j-1}$ for $ i \leq j \leq n$. 
  \[
    % \dD_k(x_{1:n})
    % = \dQ_1(x_1),\dQ_{2|1}(x_2|x_1) \cdots \dQ_{k|k-1}(x_k|x_{k-1}) \dP_{k+1|k}(x_{k+1}|x_k) \cdots \dP_{n|n-1}(x_n|x_{n-1})
    % = \dQ_{1:k} (x_{1:k}) \dP_{k+1:n|k} (x_{k+1:n}|x_k)
    \dD_k(x_1,\dots,x_n)
    = \dQ_{1:k} (x_1,\dots,x_k) \dP_{k+1:n|k} (x_{k+1},\dots,x_n|x_k)
  \]
  Thus $\dD_n = \dQ$.
  By the triangle inequality, 
  \[
    \TVD(\dP,\dQ) 
    \leq \sum_{k=1}^n\TVD(\dD_{k-1},\dD_k) 
    \leq n \max_{1\leq k\leq n} \TVD(\dD_{k-1},\dD_k) \,.
  \]
  For any $1 \leq k \leq n$, it holds that 
  \[
  \begin{aligned}
    \TVD(\dD_{k-1},\dD_k) \leq \TVD(\dD_{k-1},\dP) &+ \TVD(\dD_k, \dP) \\
    \qquad = \TVD(\dQ_{1:k-1},\dP_{1:k-1}) &+ \TVD(\dQ_{1:k},\dP_{1:k})
    \leq 2\TVD(\dP,\dQ).
  \end{aligned}
  \]
  Therefore, $\LBTV$  satisfies the requirement if it is defined as
  \[
    \LBTV \deq \frac12 \max_{1 \leq k \leq n} \TVD(\dD_{k-1},\dD_k).
  \]

  Note that, $\LBTV$ can be computed in time $O(q^2n)$ because
  \begin{align*}
    \TVD(\dD_{k-1},\dD_k) 
    = \TVD(\dQ_{k-1} \dP_{k|k-1},\dQ_{k-1} \dQ_{k|k-1})
  \end{align*}
  and we can compute all  $\dQ_{k-1} \dP_{k|k-1},\dQ_{k-1} \dQ_{k|k-1}$
  in $O(q^2n)$ time.
  (Here $\dQ_k$ denote the $k$-th marginal distribution of $\dQ$,
  which can be recursively computed by $\dQ_k \deq \dQ_{k|k-1} \dQ_{k-1}$.)
\end{proof}

\subsection{Analysis of our Algorithm}

In this section,
we analyze Algorithm~\ref{alg:markov} and prove Theorem~\ref{thm:markov}.

First analyze the time complexity.
Let $N = O(\frac{n}{\varepsilon} \log (\frac{n}{\varepsilon \TVD(\dP,\dQ)}))$ denote an upper bound
of the support size of every sparsified ratios generated by \textsf{Sparsify}.
% Let $N = O(\frac{n}{\varepsilon} \log (\frac{n}{\varepsilon \LBTV}))$ bound the support size of the sparsified ratios generated by \textsf{Sparsify}.
% As we will later show $\LBTV = \Omega(\frac{\TVD(\dP,\dQ)}{n})$,
% every sparsified ratio's support size is bounded by $N = O(\frac{n}{\varepsilon} \log (\frac{n}{\varepsilon \TVD(\dP,\dQ)}))$.
The support size of each concatenation ratio $\dR'\MK[x]{k:n}{k-1}$ is bounded by $Nq$.
The complexity is dominated by $nq$ calls to \textsf{Sparsify} and $nq$ calls to \textsf{Concatenate}.
Each \textsf{Sparsify} call takes at most $O(Nq)$ time,
each \textsf{Concatenate} call takes at most $O(Nq \log q)$ time.
Thus the total time complexity is bound by $O(Nnq^2\log q)$.

Let $\dR_{1:n}$, $\dR\MK{k:n}{k-1}$ denote the actual ratio and conditional ratio
\[
  \dR_{1:n} \deq ( \dP_{1:n} \| \dQ_{1:n} ),
  \qquad
  \dR\MK{k:n}{k-1} \deq
  \bigl( \dP\MK{k:n}{k-1} \bigm\| \dQ\MK{k:n}{k-1} \bigr)
   \,.
\]
% The correctness of our algorithm follow from the fact that
Our algorithm recursively computes a good approximation of $\dR\MK{k:n}{k-1}$ for smaller and smaller $k$.
In the end, the our algorithm computes $\dR_{1:n}'$, which is an approximation of $\dR_{1:n}$,
satisfying $\dR_{1:n}' \leq \dR_{1:n}$ and $\MTVD(\dR_{1:n}', \dR_{1:n}) \leq \frac12\varepsilon\TVD(\dP,\dQ)$.
Thus $\widehat\Delta = \TVD(\dR_{1:n}')$ is a sufficiently good approximation of $\TVD(\dP,\dQ)$.

In the proof, we use the conventional notation $\dP_k$ to denote the marginal distribution of the $k$-th coordinate of $\dP$.
It can be recursively defined by $\dP_k = \dP\MK{k}{k-1}\dP_{k-1}$.
The conventional notation $\dP_{t|k}$ is a Markov kernel specifying the distribution of the $t$-th coordinate conditioning on the $k$-th coordinate.

\begin{claim}
  For every $1<k\leq n$,
  the conditional ratio $\dR'\MK{k:n}{k-1}$ satisfies
  $\dR'\MK[x]{k:n}{k-1} \leq \dR\MK[x]{k:n}{k-1}$ for all $x\in[q]$.
  The ratio $\dR'_{1:n}$  satisfies $\dR'_{1:n} \leq \dR_{1:n}$.
\end{claim}

\begin{proof}
  There is an inductive proof from larger $k$ to smaller $k$.
  The base case is when $k = n$.
  The base case follows trivially from the fact that $\dR'\MK{n:n}{n-1} = \dR\MK{n:n}{n-1}$.

  For smaller $k$,
  by the inductive hypothesis, $\dR'\MK[x]{k+1:n}{k} \leq \dR\MK[x]{k+1:n}{k}$ for each $x\in[q]$.
  Order is preserved by sparsification (Lemma~\ref{lem:sparsify}), so 
  $
    \tilde\dR\MK[x]{k+1:n}{k} 
    \leq \dR'\MK[x]{k+1:n}{k}
    \leq \dR\MK[x]{k+1:n}{k} 
  $.
  Order is also preserved by concatenation (Lemma~\ref{lem:concat-order}), so
  \begin{multline*}
    \dR'\MK[x]{k:n}{k-1} 
    = \textsf{Concatenate}(\dP\MK[x]{k}{k-1},~ \dQ\MK[x]{k}{k-1},~ \tilde\dR\MK{k+1:n}{k}) \\
    \leq \textsf{Concatenate}(\dP\MK[x]{k}{k-1},~ \dQ\MK[x]{k}{k-1},~ \dR\MK{k+1:n}{k})
    =\dR\MK[x]{k:n}{k-1}  \,.
  \end{multline*}

  In the special case when $k=1$,
  by the same argument (Lemma~\ref{lem:sparsify}),
  $\tilde\dR\MK[x]{2:n}{1} \leq \dR\MK[x]{2:n}{1}$.
  Then by a similar argument about concatenation (Lemma~\ref{lem:concat-order}),
  \[
    \dR'_{1:n} 
    = \textsf{Concatenate}(\dP_1,\, \dQ_1,\, \tilde\dR\MK{2:n}{1}) 
    \leq \textsf{Concatenate}(\dP_1,\, \dQ_1,\, \dR\MK{2:n}{1})
    =\dR_{1:n}  \,.
    \qedhere
  \]
\end{proof}

\begin{claim}
  The ratio $\dR'_{1:n}$ in Algorithm~\ref{alg:markov} satisfies 
  $\MTVD(\dR'_{1:n},\dR_{1:n}) \leq \frac{n-1}{2n}\varepsilon \TVD(\dP,\dQ)$. 
\end{claim}

% The proof above is quite dense.
% We also provide a more intuitive proof using the hybrid argument.
\noindent
We present a rather intuitive proof here.
A more formal proof is given in \Cref{sec:markov-alter}.

\newcommand\hb{i}
\begin{proof}
  Image that there are $n$ hybrid worlds, numbered by $\hb=1,\dots,n$.
  % The 1st hybrid world is the same as the real word.
  In the $\hb$-th hybrid world, we modify how Algorithm~\ref{alg:markov} works.
  % the algorithm will skip the sparsification step in the main loop for every index $k<\hb$.
  \begin{description}
    \item[The $\hb$-th hybrid world] 
      consider a modified algorithm that skips the sparsification step in the main loop for every index $k<\hb$.
      In other words, let $\tilde\dR\MK{k+1:n}{k} = \dR'\MK{k+1:n}{k}$ for every $k<\hb$.
      To avoid confusion, 
      we use $\dR^{(\hb)}\MK{k+1:n}{k}$ to denote this value for every $k<\hb$,
      we use $\dR^{(\hb)}_{1:n}$ to denote the value of $\dR'_{1:n}$ in the $\hb$-th hybrid world.
      % The modified algorithm can be highly inefficient.
    % \item[The 1st hybrid world] is the same as the real world.
  \end{description}

  The 1st hybrid world is identical to the real world, thus $\dR^{(1)}_{1:n} = \dR'_{1:n}$.
  The last hybrid world is ``ideal'' in the sense that there is no error introduced by sparsification,
  thus $\dR^{(n)}_{1:n} = \dR_{1:n}$.
  To bound the MTV distance between $\dR'_{1:n}$ and $\dR_{1:n}$,
  it suffice to prove  for every $1\leq \hb < n$,
  \[
    \MTVD(\dR^{(\hb)}_{1:n},\dR^{(\hb+1)}_{1:n})
    \leq \frac{1}{2n}\varepsilon \TVD(\dP,\dQ) \,.
  \]

  Comparing the $\hb$-th and $(\hb+1)$-th hybrid worlds,
  the only difference is whether the sparsification step is skipped when $k=\hb$.
  % \begin{itemize}
  %   \item In the $(\hb+1)$-th hybrid world,
  %     $\dR^{(\hb+1)}\MK{\hb+1:n}{\hb} = \dR'\MK{\hb+1:n}{\hb}$.
  %   \item In the $\hb$-th hybrid world,
  %     $\tilde\dR\MK[x]{\hb+1:n}{\hb}$ is the sparsification of $\dR'\MK[x]{\hb+1:n}{\hb}$. 
  %     By Lemma~\ref{lem:sparsify},
  %     \[
  %       \MTVD(\tilde\dR\MK[x]{\hb+1:n}{\hb}, \dR^{(\hb+1)}\MK[x]{\hb+1:n}{\hb})
  %       \leq \frac{\varepsilon}{4n} \LBTV + \frac{\varepsilon}{8n} \TVD(\dR^{(\hb+1)}\MK[x]{\hb+1:n}{\hb})
  %       \leq \frac{\varepsilon}{4n} \LBTV + \frac{\varepsilon}{8n} \TVD(\dR\MK[x]{\hb+1:n}{\hb}) \,.
  %     \]
  % \end{itemize}
  In the $(\hb+1)$-th hybrid world,
  $\dR^{(\hb+1)}\MK{\hb+1:n}{\hb} = \dR'\MK{\hb+1:n}{\hb}$.
  In the $\hb$-th hybrid world,
  $\tilde\dR\MK[x]{\hb+1:n}{\hb}$ is the sparsification of $\dR'\MK[x]{\hb+1:n}{\hb}$. 
  By Lemma~\ref{lem:sparsify},
  \[
    \MTVD(\tilde\dR\MK[x]{\hb+1:n}{\hb}, \dR^{(\hb+1)}\MK[x]{\hb+1:n}{\hb})
    \leq \frac{\varepsilon}{4n} \LBTV + \frac{\varepsilon}{8n} \TVD(\dR^{(\hb+1)}\MK[x]{\hb+1:n}{\hb})
    \leq \frac{\varepsilon}{4n} \LBTV + \frac{\varepsilon}{8n} \TVD(\dR\MK[x]{\hb+1:n}{\hb}) \,.
  \]
  So there exist Markov kernels $\dP^{(\hb)*}\MK{Y}{\hb}$, $\dQ^{(\hb)*}\MK{Y}{\hb}$, 
  $\dP^{(\hb+1)**}\MK{Y}{\hb}$, $\dQ^{(\hb+1)**}\MK{Y}{\hb}$
  from $[q]$ to a sample space $\Omega^{(\hb)}_Y$ satisfying
  \begin{gather}
    \tilde\dR\MK{\hb+1:n}{\hb}     = \bigl(\dP^{(\hb)*}\MK{Y}{\hb} \bigm\| \dQ^{(\hb)*}\MK{Y}{\hb} \bigr) \,, \qquad
    \dR^{(\hb+1)}\MK{\hb+1:n}{\hb} = \bigl(\dP^{(\hb+1)**}\MK{Y}{\hb} \bigm\| \dQ^{(\hb+1)**}\MK{Y}{\hb} \bigr) \,, 
    \label{eq:artificial}
    \\
    \max\Bigl( 
      \TVD\Bigl(\dP^{(\hb)*}\MK[x]{Y}{\hb}, \dP^{(\hb+1)**}\MK[x]{Y}{\hb}\Bigr) ,\;
      \TVD\Bigl(\dQ^{(\hb)*}\MK[x]{Y}{\hb}, \dQ^{(\hb+1)**}\MK[x]{Y}{\hb}\Bigr)
    \Bigr) \leq \frac{\varepsilon}{4n} \LBTV + \frac{\varepsilon}{8n} \TVD(\dR\MK[x]{\hb+1:n}{\hb}) \,. 
    \label{eq:artificial-MTVD}
  \end{gather}
  
  Inspired by these distributions,
  we define some extra artificial hybrid worlds: 
  \begin{description}
    \item[The $\hb$* hybrid world] 
      is a truncated version of the $\hb$-th hybrid word.
      In this hybrid worlds,
      the algorithm sets $\tilde\dR\MK{\hb+1:n}{\hb}$ according to \eqref{eq:artificial}
      and the rest of the computation is the same as the $\hb$-th hybrid word.
      In this hybrid world, the algorithm will also compute $\dR^{(\hb)}_{1:n}$.

      By its definition, the algorithm computes the exact ratio between
      \[
        \dP_1\dP\MK21\dots\dP\MK{\hb}{\hb-1}\dP^{(\hb)*}\MK{Y}{\hb} \text{ and } 
        \dQ_1\dQ\MK21\dots\dQ\MK{\hb}{\hb-1}\dQ^{(\hb)*}\MK{Y}{\hb} \,.
      \]

    \item[The $(\hb+1)$** hybrid world] 
      is a truncated version of the $(\hb+1)$-th hybrid word.
      In this hybrid worlds,
      the algorithm sets $\dR^{(\hb+1)}\MK{\hb+1:n}{\hb}$ according to \eqref{eq:artificial}
      and the rest of the computation is the same as the $\hb$-th hybrid word.
      In this hybrid world, the algorithm will also compute $\dR^{(\hb+1)}_{1:n}$.

      By its definition, the algorithm computes the exact ratio between
      \[
        \dP_1\dP\MK21\dots\dP\MK{\hb}{\hb-1}\dP^{(\hb+1)**}\MK{Y}{\hb} \text{ and } 
        \dQ_1\dQ\MK21\dots\dQ\MK{\hb}{\hb-1}\dQ^{(\hb+1)**}\MK{Y}{\hb} \,.
      \]
  \end{description}
  Therefore, 
  \begin{multline*}    
    \MTVD(\dR^{(\hb)}_{1:n}, \dR^{(\hb+1)}_{1:n})
    \leq \max\Bigl(
    \TVD\Bigl( \dP_1\dP\MK21\dots\dP\MK{\hb}{\hb-1}\dP^{(\hb)*}\MK{Y}{\hb},\;
        \dP_1\dP\MK21\dots\dP\MK{\hb}{\hb-1}\dP^{(\hb+1)**}\MK{Y}{\hb} \Bigr),  \\
    \TVD\Bigl( \dQ_1\dQ\MK21\dots\dQ\MK{\hb}{\hb-1}\dQ^{(\hb)*}\MK{Y}{\hb},\;
        \dQ_1\dQ\MK21\dots\dQ\MK{\hb}{\hb-1}\dQ^{(\hb+1)**}\MK{Y}{\hb} \Bigr) \Bigr) \,.
  \end{multline*}
  The right-hand side of the inequality is bounded by 
  % (the other term is bounded symmetrically)
  \[
  \begin{aligned}[b]
    & \TVD\Bigl( \dP_1\dP\MK21\dots\dP\MK{\hb}{\hb-1}\dP^{(\hb)*}\MK{Y}{\hb},\;
        \dP_1\dP\MK21\dots\dP\MK{\hb}{\hb-1}\dP^{(\hb+1)**}\MK{Y}{\hb} \Bigr) \\
    &= \TVD\Bigl( \dP_i\dP^{(\hb)*}\MK{Y}{\hb},\; \dP_i\dP^{(\hb+1)**}\MK{Y}{\hb} \Bigr) \\
    &= \Ex_{x\sim\dP_i} \Bigl[ \TVD\Bigl( \dP^{(\hb)*}\MK[x]{Y}{\hb},\; \dP^{(\hb+1)**}\MK[x]{Y}{\hb} \Bigr) \Bigr] \\
    &\leq \Ex_{x\sim\dP_i} \Bigl[ \frac{\varepsilon}{4n} \LBTV + \frac{\varepsilon}{8n} \TVD(\dR\MK[x]{\hb+1:n}{\hb}) \Bigr] \\
    &= \frac{\varepsilon}{4n} \LBTV + \frac{\varepsilon}{8n} \Ex_{x\sim\dP_\hb} \Bigl[  \TVD(\dR\MK[x]{\hb+1:n}{\hb}) \Bigr] \\
    &= \frac{\varepsilon}{4n} \LBTV + \frac{\varepsilon}{8n} \Ex_{x\sim\dP_\hb} \Bigl[  \TVD(\dP\MK[x]{\hb+1:n}{\hb}, \dQ\MK[x]{\hb+1:n}{\hb}) \Bigr] \\
    &= \frac{\varepsilon}{4n} \LBTV + \frac{\varepsilon}{8n} \TVD(\dP_\hb\dP\MK{\hb+1:n}{\hb}, \dP_\hb\dQ\MK{\hb+1:n}{\hb}) \\
    &\leq \frac{\varepsilon}{2n} \TVD(\dP,\dQ) \,.
  \end{aligned}
  \]
  The first inequality symbol follows from formula \eqref{eq:artificial-MTVD}.
  The second inequality symbol relies on
  \begin{multline*}
    \TVD (\dP_i \dP\MK{i+1:n}{i}, \dP_i \dQ\MK{i+1:n}{i}) 
    \leq{} \TVD (\dP_i \dP\MK{i+1:n}{i}, \dQ_i \dQ\MK{i+1:n}{i}) + \TVD (\dQ_i \dQ\MK{i+1:n}{i}, \dP_i \dQ\MK{i+1:n}{i}) \\
    ={} \TVD (\dP_{i:n} , \dQ_{i:n} ) + \TVD (\dQ_{i}, \dP_{i} ) 
    \leq{} 2 \TVD(\dP,\dQ) \,.
    \qedhere
  \end{multline*}
  % Symmetrically, $\TVD\Bigl( \dQ_1\dQ\MK21\dots\dQ\MK{\hb}{\hb-1}\dQ^{(\hb)*}\MK{Y}{\hb},\;
  %       \dQ_1\dQ\MK21\dots\dQ\MK{\hb}{\hb-1}\dQ^{(\hb+1)**}\MK{Y}{\hb} \Bigr)$ is also bounded by $\frac{\varepsilon}{2n} \TVD(\dP,\dQ)$.
\end{proof}

\section*{Acknowledgement}
%We thank all 
Weiming Feng would like to thank Dr. Heng Guo for the helpful discussions.
Tianren Liu would like to thank Prof.\ Yury Polyanskiy for introducing us to the concepts of deficiency and Le Cam's distance.

\bibliographystyle{alpha}
\bibliography{refs.bib}

\appendix
\section{Another Analysis of our Algorithm in the Markov Chain Setting}
\label{sec:markov-alter}
% \subsection{Region of TV Distance Pair}

% \tianren{ongoing; will be moved to appendix}

This section presents an alternative analysis for bounding the error of Algorithm~\ref{alg:markov}.
It starts by introducing a finer characterization of the ``difference'' between two ratios.
% To analyze our algorithm in the Markov chain setting, 
% we need to introduce a finer characterization of the ``difference'' between two ratios.

\begin{definition}[Region]
  For any two ratios $\dR_1,\dR_2$,
  define the feasible region of TV distance pair (or ``\emph{region}'', in short) between $\dR_1,\dR_2$,
  denoted by $\RTVD(\dR_1,\dR_2)$, as
  \[
    \RTVD(\dR_1,\dR_2)
    \deq \left\{ \Big(\TVD(\dP_1,\dP_2), \TVD(\dQ_1,\dQ_2)\Bigr)  \middle|\; \begin{aligned} &\text{distributions } \dP_1,\dQ_1,\dP_2,\dQ_2 \\&\text{such that } (\dP_1\|\dQ_1) = \dR_1,(\dP_2\|\dQ_2) = \dR_2 \end{aligned} \right\} \,.
  \]
\end{definition}

Region has many nice properties.
It is more expressive than the minimum total variation distance.
The later can be defined as the $L_1$-distance between $(0,0)$ and the region.
\[
  \MTVD(\dR_1,\dR_2) = \inf_{(\delta_1,\delta_2) \in \RTVD(\dR_1,\dR_2)} \max(\delta_1,\delta_2) \,.
\]
We remark that, the current definition of $\MTVD$ is not special.
For example, consider an alternative distance
\[
  \MTVD'(\dR_1,\dR_2) = \inf_{(\delta_1,\delta_2) \in \RTVD(\dR_1,\dR_2)} (\delta_1 + \delta_2) \,.
\]
This alternative distance, which is the $L_\infty$-distance from $(0,0)$, also works well with the paper.
% We remark that, there are alternative distances that also work well with this paper.
% For example, consider
% \[
%   \MTVD'(\dR_1,\dR_2) = \inf_{(\delta_1,\delta_2) \in \RTVD(\dR_1,\dR_2)} (\delta_1 + \delta_2) \,,
% \]
% which measures the distance between $(0,0)$ and $\RTVD(\dR_1,\dR_2)$ in $L_1$-norm.
An early draft of this paper uses $\MTVD'$ to analyze our algorithms.

% Region has many nice properties:
% containing $(1,1)$; convex; closed, if the support size is finite.
For this paper, it is sufficient to have the following property.
The proof is deferred to \Cref{sec:proof}.
% which has been implicitly proved in the proof of \Cref{lem-metric}.
%
\begin{lemma}[Triangle Inequality]\label{lem:region-triangle}
  For any ratios $\dR_1,\dR_2,\dR$
  such that $\RTVD(\dR_1,\dR) \leq_\exists (\alpha_1,\beta_1)$ 
  and $\RTVD(\dR_2,\dR) \leq_\exists (\alpha_2,\beta_2)$,
  it holds that  
  $\RTVD(\dR_1,\dR_2) \leq_\exists (\alpha_1+ \alpha_2,\beta_1+\beta_2)$.
\end{lemma}

% This proposition is implicitly proved in the proof of \Cref{lem-metric}.

The description of this lemma uses a convenient notation $\leq_\exists$,
defined as follows.
% The proposition is presented using some convenient notations that will simplify our analysis.
For any two points $(v_1,v_2)$, $(w_1,w_2)$ in the Euclidean plane,
we say $(v_1,v_2) \leq (w_1,w_2)$ if $v_1\leq w_1$ and $v_2\leq w_2$.
Apparently, this defines an order relation.
Let $\cR$ be a point set in the Euclidean plane,
we say  $\cR \leq_\exists (v_1,v_2)$ if there exists $(u_1,u_2) \in\cR$ such that $(u_1,u_2) \leq (v_1,v_2)$.

Using the new concept of region, we gives a fine characterization of how the error passes through the concatenation process.

\begin{lemma}\label{lem:concat-MTVD}
  Given conditional ratios $\dR\MK YX, \dR'\MK YX$ over $\Omega_X$,
  for any mappings $\alpha, \beta: \Omega_X\to[0,\infty)$ such that $\RTVD(\dR\MK[x]YX, \dR'\MK[x]YX) \leq_\exists (\alpha(x),\beta(x))$ for every $x\in\Omega_X$,
  for any distributions $\dP_X,\dQ_X$ over $\Omega_X$,
  it holds that
  \[
    \RTVD(\textnormal{\textsf{Concatenate}} (\dP_X,\dQ_X,\dR\MK YX), \textnormal{\textsf{Concatenate}} (\dP_X,\dQ_X,\dR'\MK YX) )
    \leq_\exists
    \Bigl(\Ex_{X\sim\dP_X} [ \alpha(X) ], \Ex_{X\sim\dQ_X} [ \beta(X) ]\Bigr)  \,.
  \]
\end{lemma}

\begin{proof}
  Since $\RTVD(\dR\MK[x]YX, \dR'\MK[x]YX) \leq_\exists  (\alpha(x),\beta(x))$ for every $x\in\Omega_X$,
  % there exist a sample space $\Omega^{(x)}_Y$,
  % distributions $\dP^{(x)},\dQ^{(x)},\dP'^{(x)},\dQ'^{(x)}$ over $\Omega^{(x)}_Y$
  % such that 
  % \[
  %   \dR\MK[x]YX = (\dP^{(x)} \| \dQ^{(x)}) \,,\qquad
  %   \dR'\MK[x]YX = (\dP'^{(x)} \| \dQ'^{(x)}) \,.
  % \]
  there exist a sample space $\Omega_Y$,
  Markov kernels $\dP\MK YX,\dP'\MK YX,\dQ\MK YX,\dQ'\MK YX$ from $\Omega_X$ to $\Omega_Y$
  % and Markov kernels $\kappa_x$ from $\Omega_Y$ to $\Omega_Y$ for each $x\in\Omega_X$
  such that 
  \[
  \begin{gathered}
    \dR\MK YX = (\dP\MK YX \| \dQ\MK YX) \,, \qquad \dR'\MK YX = (\dP'\MK YX \| \dQ'\MK YX) \,, \\
    \TVD(\dP\MK[x]YX, \dP'\MK[x]YX) \leq \alpha(x), \quad \TVD(\dQ\MK[x]YX, \dQ'\MK[x]YX) \leq \beta(x) ~\text{ for every }x\in\Omega_X \,.
  \end{gathered}
  \]
  The two ratios in the wanted inequality equal $(\dP_X\dP\MK YX  \| \dQ_X\dQ\MK YX )$ and $(\dP_X\dP'\MK YX \| \dQ_X\dQ'\MK YX)$ respectively. 
  Therefore\footnote{The expectation in the 3rd line of the following formula is written using an inconsistent but intuitive notation $x\sim\dP_X$.  
  If using the same notation as the rest of the paper, it should be written as $\Ex_{X\sim\dP_X} \bigl[\TVD(\dP\MK[X]YX, \dP'\MK[X]YX)\bigr]$.}
  \[
  \begin{aligned}[b]
    &\RTVD(\textnormal{\textsf{Concatenate}} (\dP_X,\dQ_X,\dR\MK YX), \textnormal{\textsf{Concatenate}} (\dP_X,\dQ_X,\dR'\MK YX) ) \\
    &\leq_\exists 
    \bigl(\TVD(\dP_X\dP\MK YX,\dP_X\dP'\MK YX), ~ \TVD( \dQ_X\dQ\MK YX , \dQ_X\dQ'\MK YX)\bigr) \\
    % &=
    % \max\Bigl(\sum_{x\in\Omega_X} \dP_X(x) \TVD(\dP\MK[x]YX, \dP'\MK[x]YX), \sum_{x\in\Omega_X}  \dQ_X(x) \TVD( \dQ\MK[x]YX , \dQ'\MK[x]YX) \Bigr) \\
    &= 
    \Bigl(\Ex_{x\sim\dP_X} \bigl[\TVD(\dP\MK[x]YX, \dP'\MK[x]YX)\bigr], \Ex_{x\sim\dQ_X} \bigl[\TVD( \dQ\MK[x]YX , \dQ'\MK[x]YX)\bigr]\Bigr) \\
    &\leq
    \Bigl(\Ex_{X\sim\dP_X} [ \alpha(X) ], \Ex_{X\sim\dQ_X} [ \beta(X) ]\Bigr)  \,.
  \end{aligned}
  \qedhere
  \]
\end{proof}

Now we are ready to present an alternative error analysis of Algorithm~\ref{alg:markov}.

\begin{claim}
  For every $1<k\leq n$,
  the intermediate conditional ratio $\dR'\MK{k:n}{k-1}$ satisfies
  \begin{multline*}
    \RTVD (\dR'\MK[x]{k:n}{k-1}, \dR\MK[x]{k:n}{k-1})
    \leq_\exists
    \biggl(\frac{n-k}{4n}\cdot \varepsilon \LBTV + \frac{\varepsilon}{8n} \sum_{t=k}^{n-1}\Ex_{z\sim \dP\MK[x]t{k-1}}\bigl[ \TVD(\dR\MK[z]{t+1:n}{t}) \bigr], \\
    \frac{n-k}{4n}\cdot \varepsilon \LBTV + \frac{\varepsilon}{8n} \sum_{t=k}^{n-1}\Ex_{z\sim \dQ\MK[x]t{k-1}}\bigl[ \TVD(\dR\MK[z]{t+1:n}{t}) \bigr]
    \biggr) \,,
  \end{multline*}
  for all $x\in[q]$.
  The ratio $\dR'_{1:n}$ satisfies 
  $\MTVD(\dR'_{1:n},\dR_{1:n}) \leq \frac{n-1}{2n}\varepsilon \TVD(\dP,\dQ)$. 
\end{claim}

\begin{proof}
  There is an inductive proof from larger $k$ to smaller $k$.
  The base case is when $k = n$.
  The base case follows trivially from the fact that $\dR'\MK{n:n}{n-1} = \dR\MK{n:n}{n-1}$.

  For every smaller $k$, by Lemma~\ref{lem:sparsify},
  the error introduced by sparsifying $\dR'\MK[x]{k+1:n}{k}$ is bounded by
  \[
    \MTVD (\tilde\dR\MK[x]{k+1:n}{k}, \dR'\MK[x]{k+1:n}{k})
    \leq \frac{\varepsilon}{4n} \LBTV + \frac{\varepsilon}{8n} \TVD(\dR'\MK[x]{k+1:n}{k})
    \leq \frac{\varepsilon}{4n} \LBTV + \frac{\varepsilon}{8n} \TVD(\dR\MK[x]{k+1:n}{k}) \,,
  \]
  which means
  \[
    \RTVD (\tilde\dR\MK[x]{k+1:n}{k}, \dR'\MK[x]{k+1:n}{k})
    \leq_\exists \Bigl(\frac{\varepsilon}{4n} \LBTV + \frac{\varepsilon}{8n} \TVD(\dR\MK[x]{k+1:n}{k}),\frac{\varepsilon}{4n} \LBTV + \frac{\varepsilon}{8n} \TVD(\dR\MK[x]{k+1:n}{k})\Bigr)  \,.
  \]
  Combining with the inductive hypothesis that bounds $\RTVD (\dR'\MK[x]{k+1:n}{k}, \dR\MK[x]{k+1:n}{k})$ by
  \begin{multline*}
    \RTVD (\tilde\dR\MK[x]{k+1:n}{k}, \dR\MK[x]{k+1:n}{k})
    \leq_\exists 
    \biggl(\frac{n-k - 1}{4n}\cdot \varepsilon \LBTV + \frac{\varepsilon}{8n} \sum_{t=k+1}^{n-1}\Ex_{z\sim \dP\MK[x]t{k}}\bigl[ \TVD(\dR\MK[z]{t+1:n}{t}) \bigr], \\
    \frac{n-k - 1}{4n}\cdot \varepsilon \LBTV + \frac{\varepsilon}{8n} \sum_{t=k+1}^{n-1}\Ex_{z\sim \dQ\MK[x]t{k}}\bigl[ \TVD(\dR\MK[z]{t+1:n}{t}) \bigr]
    \biggr) \,,
  \end{multline*}
  using the triangle inequality (\Cref{lem:region-triangle})
  \begin{multline}
  \label{eq:markov-analysis-triangle}
    \RTVD (\dR'\MK[x]{k+1:n}{k}, \dR\MK[x]{k+1:n}{k})
    \leq_\exists
    \biggl(\frac{n-k}{4n}\cdot \varepsilon \LBTV + \frac{\varepsilon}{8n} \sum_{t=k}^{n-1}\Ex_{z\sim \dP\MK[x]t{k}}\bigl[ \TVD(\dR\MK[z]{t+1:n}{t}) \bigr], \\
    \frac{n-k}{4n}\cdot \varepsilon \LBTV + \frac{\varepsilon}{8n} \sum_{t=k}^{n-1}\Ex_{z\sim \dQ\MK[x]t{k}}\bigl[ \TVD(\dR\MK[z]{t+1:n}{t}) \bigr]
    \biggr) \,.
  \end{multline}
  Note that, in the above formula, $\dP\MK[x]kk$ is the degenerated distribution that only samples $x$.
  Then the inductive step is concluded by Lemma~\ref{lem:concat-MTVD}
  \begin{align}
    \RTVD (\dR'\MK[x]{k:n}{k-1}, \dR\MK[x]{k:n}{k-1})
    &
    \begin{aligned}[t]
    {}\leq_\exists
      \biggl(
      \Ex_{y\sim\dP\MK[x]{k}{k-1}} \Bigl[ \frac{n-k}{4n}\cdot \varepsilon \LBTV + \frac{\varepsilon}{8n} \sum_{t=k}^{n-1}\Ex_{z\sim \dP\MK[y]t{k}}\bigl[ \TVD(\dR\MK[z]{t+1:n}{t}) \bigr] \Bigr], \qquad \\
      \Ex_{y\sim\dQ\MK[x]{k}{k-1}} \Bigl[ \frac{n-k}{4n}\cdot \varepsilon \LBTV + \frac{\varepsilon}{8n} \sum_{t=k}^{n-1}\Ex_{z\sim \dQ\MK[y]t{k}}\bigl[ \TVD(\dR\MK[z]{t+1:n}{t}) \bigr] \Bigr]
      \biggr) 
    \end{aligned} \nonumber\\
    &
    \begin{aligned}
    {}=
      \biggl(
       \frac{n-k}{4n}\cdot \varepsilon \LBTV + \frac{\varepsilon}{8n} \sum_{t=k}^{n-1}\Ex_{y\sim\dP\MK[x]{k}{k-1}} \Bigl[ \Ex_{z\sim \dP\MK[y]t{k}}\bigl[ \TVD(\dR\MK[z]{t+1:n}{t}) \bigr] \Bigr], \qquad \\
      \frac{n-k}{4n}\cdot \varepsilon \LBTV + \frac{\varepsilon}{8n} \sum_{t=k}^{n-1}\Ex_{y\sim\dQ\MK[x]{k}{k-1}} \Bigl[ \Ex_{z\sim \dQ\MK[y]t{k}}\bigl[ \TVD(\dR\MK[z]{t+1:n}{t}) \bigr] \Bigr]
      \biggr) 
    \end{aligned} \nonumber\\
    &
    \begin{aligned}[b]
    {}=
      \biggl(
       \frac{n-k}{4n}\cdot \varepsilon \LBTV + \frac{\varepsilon}{8n} \sum_{t=k}^{n-1} \Ex_{z\sim \dP\MK[x]t{k-1}}\bigl[ \TVD(\dR\MK[z]{t+1:n}{t}) \bigr] , \qquad \\
      \frac{n-k}{4n}\cdot \varepsilon \LBTV + \frac{\varepsilon}{8n} \sum_{t=k}^{n-1} \Ex_{z\sim \dQ\MK[x]t{k-1}}\bigl[ \TVD(\dR\MK[z]{t+1:n}{t}) \bigr] 
      \biggr) \,.
    \end{aligned}
    \label{eq:markov-analysis-concat}
  \end{align}

  In the special case when $k=1$,
  by the same analysis as the first half of the inductive step,
  formula~\eqref{eq:markov-analysis-triangle} holds.
  Then by an argument mostly similar\footnote{The argument will be exactly the same, if we define degenerated distributions $\dP_0,\dQ_0$ over a size-1 sample space and define Markov kernels $\dP\MK10,\dQ\MK10$ such that as $\dP\MK10 \dP_0 = \dP_1$ and $\dQ\MK10 \dQ_0 = \dQ_1$.} to \eqref{eq:markov-analysis-concat},
  \[
    \RTVD (\dR'_{1:n}, \dR_{1:n})
    \leq_\exists
    \begin{aligned}[t]
      \biggl(
       \frac{n-1}{4n}\cdot \varepsilon \LBTV + \frac{\varepsilon}{8n} \sum_{t=1}^{n-1} \Ex_{z\sim \dP_t}\bigl[ \TVD(\dR\MK[z]{t+1:n}{t}) \bigr] , \qquad \\
      \frac{n-1}{4n}\cdot \varepsilon \LBTV + \frac{\varepsilon}{8n} \sum_{t=1}^{n-1} \Ex_{z\sim \dQ_t}\bigl[ \TVD(\dR\MK[z]{t+1:n}{t}) \bigr] 
      \biggr) \,.
    \end{aligned}
  \]

  % \[
  % \begin{aligned}
  %   & \Ex_{z\sim \dQ_t}\bigl[ \TVD(\dR\MK[z]{t+1:n}{t}) \bigr]  \\
  %   ={}& \Ex_{z\sim \dQ_t}\bigl[ \TVD (\dP\MK[z]{t+1:n}{t}, \dQ\MK[z]{t+1:n}{t}) \bigr]  \\
  %   ={}& \TVD (\dQ_t \dP\MK{t+1:n}{t}, \dQ_t \dQ\MK{t+1:n}{t}) \\
  %   \leq{}& \TVD (\dQ_t \dP\MK{t+1:n}{t}, \dP_t \dP\MK{t+1:n}{t}) + \TVD (\dP_t \dP\MK{t+1:n}{t}, \dQ_t \dQ\MK{t+1:n}{t}) \\
  %   ={}& \TVD (\dQ_t , \dP_t ) + \TVD (\dP_{t:n}, \dQ_{t:n} ) \\
  %   \leq{}& 2 \TVD(\dP,\dQ)
  % \end{aligned}
  % \]
  For each $1\leq t<n$,
  we have
  \begin{multline*}
    \Ex_{z\sim \dQ_t}\bigl[ \TVD(\dR\MK[z]{t+1:n}{t}) \bigr]  
    = \Ex_{z\sim \dQ_t}\bigl[ \TVD (\dP\MK[z]{t+1:n}{t}, \dQ\MK[z]{t+1:n}{t}) \bigr]  \\
    = \TVD (\dQ_t \dP\MK{t+1:n}{t}, \dQ_t \dQ\MK{t+1:n}{t}) 
    \leq{} \TVD (\dQ_t \dP\MK{t+1:n}{t}, \dP_t \dP\MK{t+1:n}{t}) + \TVD (\dP_t \dP\MK{t+1:n}{t}, \dQ_t \dQ\MK{t+1:n}{t}) \\
    ={} \TVD (\dQ_t , \dP_t ) + \TVD (\dP_{t:n}, \dQ_{t:n} ) 
    \leq{} 2 \TVD(\dP,\dQ) \,.
  \end{multline*}
  Symmetrically, $\Ex_{z\sim \dP_t}\bigl[ \TVD(\dR\MK[z]{t+1:n}{t}) \bigr] \leq 2 \TVD(\dP,\dQ)$.
  Therefore, 
  \begin{align*}  
    \RTVD (\dR'_{1:n}, \dR_{1:n})
    &\leq_\exists
    \begin{aligned}[t]
      \biggl(
       \frac{n-1}{4n}\cdot \varepsilon \LBTV + \frac{\varepsilon}{8n} (n-1) \cdot 2\TVD(\dP,\dQ) , ~\text{the same}
      \biggr) \,.
    \end{aligned}  \\
    \MTVD (\dR'_{1:n}, \dR_{1:n})
    &\leq \frac{n-1}{4n}\cdot \varepsilon \LBTV + \frac{\varepsilon}{8n} (n-1) \cdot 2\TVD(\dP,\dQ) \\
    &\leq \frac{n-1}{2n}\cdot \varepsilon \TVD(\dP,\dQ) \,.
    \qedhere
  \end{align*}
\end{proof}

\section{Deferred Proofs}
% \section{Deferred Proofs for the Minimum TV Distance}
\label{sec:proof}
\label{sec-proof-mtvd}

\begin{proof}[Proof of \Cref{lem:inf}]
  The key step of the proof is 
  to enforce the search space of $(\dP_1,\dQ_1,\dP_2,\dQ_2)$ to be finite-dimensional.
  Concretely, we construct a \emph{finite} set $\Omega$ such that we can assume w.l.o.g.~that $\Omega$ is the sample space of distributions $\dP_1,\dQ_1,\dP_2,\dQ_2$.
  That is,
  % Assume there is a \emph{finite-size} sample space $\Omega$ 
  % such that
  \begin{equation}\label{eq:finiteomega}
    \MTVD(\dR_1,\dR_2)
    =
    \inf_{\substack{\textit{discrete}~\dP_1,\dQ_1,\dP_2,\dQ_2\\\textit{in sample space }\Omega\\(\dP_1\|\dQ_1) = \dR_1 \\(\dP_2\|\dQ_2) = \dR_2}} \max\bigl( \TVD(\dP_1,\dP_2), \TVD(\dQ_1,\dQ_2) \bigr) \,.
  \end{equation}

  Assume that we have found such a \emph{finite} set $\Omega$ satisfying equation~\eqref{eq:finiteomega},
  then the infimum symbol can be replaced by minimum,
  because the search space is a compact set
  and the TV distance is a continuous function.
  Thus to finish the proof, it suffices to (explicitly) construct $\Omega$ and show \eqref{eq:finiteomega}.

  % % Let $\supp(R_1)$ denote the support of $R_1$.
  % Define $\cS_1 \deq \supp(\dR_1) \cup \{\infty\}$ and $\cS_2 \deq \supp(\dR_2) \cup \{\infty\}$.
  % The finite sample space in \eqref{eq:finiteomega} is defined as $\Omega = \cS_1 \times \cS_2$.

  The sample space $\Omega$ is defined as
  \[
    \Omega = \Bigl( \supp(\dR_1) \cup \{\infty\} \Bigr) \times \Bigl( \supp(\dR_2) \cup \{\infty\} \Bigr).
  \]
  The desired~\eqref{eq:finiteomega} is implied by the following statement:
  for any discrete distributions $\dP'_1,\dQ'_1,\dP'_2,\dQ'_2$ over any sample space $\Omega'$ such that $(\dP'_1\|\dQ'_1) = \dR_1$ and $(\dP'_2\|\dQ'_2) = \dR_2$,
  there exists distributions $\dP_1,\dQ_1,\dP_2,\dQ_2$ over $\Omega$ such that $(\dP_1\|\dQ_1) = \dR_1$, $(\dP_2\|\dQ_2) = \dR_2$ and
  % \[
  %   \max(\TVD(\dP_1,\dP_2), \TVD(\dQ_1,\dQ_2))
  %   \leq \max(\TVD(\dP'_1,\dP'_2), \TVD(\dQ'_1,\dQ'_2)).
  % \]
  \begin{equation}\label{eq:prop-data-processing}
    \TVD(\dP_1,\dP_2) \leq \TVD(\dP'_1,\dP'_2) 
    \quad\text{ and }\quad
    \TVD(\dQ_1,\dQ_2)  \leq \TVD(\dQ'_1,\dQ'_2).
  \end{equation}
  The distributions $\dP_1,\dQ_1,\dP_2,\dQ_2$ are constructed as follows.
  % Define mappings $\pi_1:\Omega'\to\cS_1$, $\pi_2:\Omega'\to\cS_2$, 
  Define mapping $\pi:\Omega'\to\Omega$ as
  \[
    \pi(z) = \Bigl(\frac{\dP'_1(z)}{\dQ'_1(z)}, \frac{\dP'_2(z)}{\dQ'_2(z)}\Bigr).
  \]
  % (Here we let a fraction equal $\infty$ if its denominator is zero.)
  Let $\dP_1,\allowbreak\dQ_1,\allowbreak\dP_2,\allowbreak\dQ_2$ be the distributions of $\pi(X_1),\allowbreak\pi(Y_1),\allowbreak\pi(X_2),\allowbreak\pi(Y_2)$
  where $X_1\sim\dP'_1,\allowbreak Y_1\sim\dQ'_1,\allowbreak X_2\sim\dP'_2,\allowbreak Y_2\sim\dQ'_2$.
  By the data-processing principle, \eqref{eq:prop-data-processing} holds.
  % \[
  %   \TVD(\dP_1,\dP_2) \leq \TVD(\dP'_1,\dP'_2) 
  %   \quad\text{ and }\quad
  %   \TVD(\dQ_1,\dQ_2)  \leq \TVD(\dQ'_1,\dQ'_2).
  % \]
  Now the only remaining task is to show $(\dP_1\|\dQ_1) = \dR_1$, $(\dP_2\|\dQ_2) = \dR_2$.
  By symmetry, it suffices to prove one of them.

  For each $(r_1,r_2) \in \Omega$
  such that $\dQ_1(r_1,r_2) > 0$ 
  (which implies $r_1\neq \infty$),
  \[
    \frac{\dP_1(r_1,r_2)}{\dQ_1(r_1,r_2)}
    = \frac{\sum_{{z\in\Omega' \text{~s.t.~} \pi(z)=(r_1,r_2)}} \dP'_1(z)}{\sum_{{z\in\Omega' \text{~s.t.~}  \pi(z)=(r_1,r_2)}} \dQ'_1(z)}
    = \frac{\sum_{{z\in\Omega' \text{~s.t.~} \pi(z)=(r_1,r_2)}} r_1 \cdot \dQ'_1(z)}{\sum_{{z\in\Omega' \text{~s.t.~}  \pi(z)=(r_1,r_2)}} \dQ'_1(z)}
    = r_1.
  \]
  For each $r\in[0,\infty)$
  \begin{multline*}
    (\dP_1\|\dQ_1)(r)
    = \Pr_{Y\gets \dQ_1}\Bigl[ \frac{\dP_1(Y)}{\dQ_1(Y)} = r \Bigr]
    = \Pr_{Y\gets \dQ_1}\Bigl[ \text{first entry of }Y = r \Bigr] \\
    = \Pr_{Y'\gets \dQ'_1}\Bigl[ \text{first entry of }\pi(Y') = r \Bigr] 
    = \Pr_{Y'\gets \dQ'_1}\Bigl[ \frac{\dP'_1(Y')}{\dQ'_1(Y')} = r \Bigr]
    = (\dP'_1\|\dQ'_1)(r).
  \end{multline*}
  Thus $(\dP_1\|\dQ_1) = (\dP'_1\|\dQ'_1) = \dR_1$.
% \tianren{unfinished proof}
\end{proof}

\begin{proof}[Proof of \Cref{lem-metric}]
  It is easy to verify that $\MTVD(\dR,\dR) = 0$ and $\MTVD(\dR_1,\dR_2) = \MTVD(\dR_2,\dR_1)$.

\medskip\noindent\textbf{Positivity:}
  Next, we verify the positivity property: $\MTVD(\dR_1,\dR_2) > 0$ if $\dR_1 \neq \dR_2$.
  If $\dR_1,\dR_2$ have finite-size supports, the positivity property is very easy to prove.
  Here we briefly present a proof without assuming finite-size supports.

  The proof requires some basic knowledge of decision theory.
  In the decision problem between two distribution $\dP$ and $\dQ$,
  a distinguisher is given a sample $x$ that is sampled from either $\dP$ or $\dQ$,
  and is supposed to guess which distribution $x$ is sampled from.
  The distinguisher can be formalized as a randomized algorithm from the sample space to $\bit$.
  It is well-known that TV distance can be defined as 
  \[
    \TVD(\dP,\dQ) = \max_{D} \Bigl| \Pr_{X\sim\dP}\bigl[D(X) \to 0\bigr] - \Pr_{X\sim\dQ}\bigl[D(X) \to 0\bigr] \Bigr|\,.
  \]
  The \emph{Neyman-Pearson region} of the decision problem, denoted by $\NPR(\dP,\dQ)$, is defined as
  \[
    \NPR(\dP,\dQ)
    \deq \Bigl\{ \bigl(\Pr_{X\sim\dP}[D(X) \to 0], \Pr_{X\sim\dQ}[D(X)\to 0]\bigr) \Bigm| D \Bigr\}
    \subseteq [0,1]\times[0,1]\,.
  \]
  % We requires the following properties of the 
  Every Neyman-Pearson region satisfies the following properties.
  \begin{itemize}
    \item The region is close and convex.
    \item The region contains $(0,0)$ and $(1,1)$.
    \item The region is centrally symmetric: $(\alpha,\beta) \in \NPR(\dP,\dQ)$ if and only if $(1-\alpha,1-\beta) \in \NPR(\dP,\dQ)$.
    \item The region is determined by the ratio, and vice verse.
      Thus we can define the Neyman-Pearson region of ratio $\dR$, denoted by $\NPR(\dR)$,
      so that $\NPR(\dR) = \NPR(\dP,\dQ)$ as long as $\dR = (\dP\|\dQ)$.
  \end{itemize}

  Given any two distinct ratios $\dR_1 \neq \dR_2$,
  we have $\NPR(\dR_1) \neq \NPR(\dR_2)$.
  Assume w.l.o.g.\ that $(\alpha,\beta) \in \NPR(\dR_1) \setminus \NPR(\dR_2)$.
  Since $\NPR(\dR_2)$ is close,
  there exists a constant $\delta>0$ such that $\max(|\alpha-\alpha'|, |\beta-\beta'|) \geq \delta$ for any $(\alpha',\beta') \in  \NPR(\dR_2)$.

  % For any $\dP_1,\dQ_1,\dP_2,\dQ_2$ such that $(\dP_1\|\dQ_1) = \dR_1$, $(\dP_2\|\dQ_2) = \dR_2$,

  For any $\dP_1,\dQ_1$ satisfying $(\dP_1\|\dQ_1) = \dR_1$,
  since $(\alpha,\beta)\in\NPR(\dP_1,\dQ_1)$,
  there exists a distinguisher $D$ such that 
  \[
    \Pr_{X\sim\dP_1}\bigl[D(X) \to 0\bigr] = \alpha, 
    \qquad
    \Pr_{X\sim\dQ_2}\bigl[D(X)\to 0\bigr] = \beta.
  \]
  For any $\dP_2,\dQ_2$ satisfying $(\dP_2\|\dQ_2) = \dR_2$, let
  \[
    (\alpha',\beta') \deq \Bigl(\Pr_{X\sim\dP_2}\bigl[D(X) \to 0\bigr],\Pr_{X\sim\dQ_2}\bigl[D(X) \to 0\bigr]\Bigr) \in \NPR(\dR_2) \,.
  \]
  Then
  \[
  \begin{aligned}
    &\max(\TVD(\dP_1,\dP_2), \TVD(\dQ_1,\dQ_2)) \\
    &\geq \max\Bigl( \Bigl| \Pr_{X\sim\dP_1}\bigl[D(X) \to 0\bigr] - \Pr_{X\sim\dP_2}\bigl[D(X) \to 0\bigr] \Bigr|
    ,~ \Bigl| \Pr_{X\sim\dQ_1}\bigl[D(X) \to 0\bigr] - \Pr_{X\sim\dQ_2}\bigl[D(X) \to 0\bigr] \Bigr| \Bigr)\\
    &= \max( |\alpha-\alpha'|, |\beta-\beta'|) \geq \delta \,.
  \end{aligned}
  \]
  Therefore $\MTVD(\dR_1,\dR_2) \geq \delta > 0$.

  % \tianren{This positivity proof assume finite-size support (so that infimum equals minimum).  But in such case there is a trivial proof.}
  % We will show that for all $P_1,Q_1,P_2,Q_2$ with $R_1 = (P_1 \| Q_1)$ and $R_2 = (P_2 \|Q_2)$, it holds that $\TVD(P_1,P_2)+\TVD(Q_1,Q_2) > 0$.
  % Fix such $P_1,Q_1,P_2,Q_2$. Suppose $P_1,P_2,Q_1,Q_2$ are defined over $\Omega$. 
  % Since $R_1 \neq  R_2$, there exists $r \in [0,\infty)$ such that $R_1(r) \neq R_2(r)$.
  % Define
  % \begin{align*}
  % 	\Omega^r_1 = \left\{ x \in \Omega \mid Q_1(x) > 0 \wedge \frac{P_1(x)}{Q_1(x)} = r \right\},
  % 	\Omega^r_2 = \left\{ x \in \Omega \mid Q_2(x) > 0 \wedge \frac{P_2(x)}{Q_2(x)} = r \right\}.
  % \end{align*}
  % If $\Omega^r_1 = \Omega^r_2$, then $\Pr_{x \gets Q_1}[x \in \Omega^r_1] = R_1(r) \neq R_2(r) = \Pr_{x \gets Q_2}[x \in \Omega^r_2]$, which implies $\TVD(Q_1,Q_2) > 0$. We now assume $\Omega^r_1 = \Omega^r_2$. W.l.o.g., we assume that there is $y \in \Omega^r_1$ such that $y \notin \Omega^r_2$ (the other case follows by symmetry). It holds that
  % \begin{align*}
  % 	Q_1(y) > 0 \wedge \frac{P_1(y)}{Q_1(y)} = r \wedge \tp{  Q_2(y) = 0 \vee \frac{P_2(y)}{Q_2(y)} \neq r }.
  % \end{align*}
  % If $Q_2(y) \neq Q_1(y)$, then $\TVD(Q_1,Q_2) \geq Q_1(y) - Q_2(y) > 0$.
  % If $Q_2(y) = Q_1(y)$, then it must hold that $Q_2(y) > 0$ and $\frac{P_2(y)}{Q_2(y)} \neq r$.
  % In this case, we have $P_2(y) \neq r Q_2(y) = r Q_1(y) = P_1(y)$, which implies $\TVD(P_1,P_2)>0$. 
  % This proves the positivity property.
  
\medskip\noindent\textbf{Triangle inequality:}
  %
  % We are going to prove the triangle inequality.  
  We will prove that for any ratios $\dR,\dR_1,\dR_2$,
  \[
    \MTVD(\dR_1,\dR_2)
    \leq \MTVD (\dR,\dR_1) + \MTVD(\dR,\dR_2).
  \]

  By the definition of MTV distance, there exist%
  \footnote{For this paper, we can assume $\dR_1,\dR$ have finite-size supports, then such distributions $\dP_1,\dQ_1,\dP_3,\dQ_3$ are guaranteed to exist (Lemma~\ref{lem:inf}). 
  % Otherwise, $\MTVD(R,R_1)$ is defined as an infimum instead of a minimum.  In such case the proof need a small modification.
  Otherwise, the proof need a small modification.} 
  distributions $\dP_1,\dQ_1,\dP_3,\dQ_3$ over sample space $\Omega_1$ 
  such that
  \[
    (\dP_1\|\dQ_1) = \dR_1, \quad
    (\dP_3\|\dQ_3) = \dR, \quad
    \max(\TVD(\dP_1,\dP_3), \TVD(\dQ_1,\dQ_3)) = \MTVD(\dR,\dR_1).
  \]
  Similarly, there exist distributions $\dP_2,\dQ_2,\dP_4,\dQ_4$ over sample space $\Omega_2$
  such that
  \[
    (\dP_2\|\dQ_2) = \dR_2, \quad
    (\dP_4\|\dQ_4) = \dR, \quad
    \max(\TVD(\dP_2,\dP_4), \TVD(\dQ_2,\dQ_4)) = \MTVD(\dR,\dR_2).
  \]
  We are going to construct distributions $\dP',\dQ',\dP'_1,\dQ'_1,\dP'_2,\dQ'_2$ 
  such that
  \begin{equation}\label{eq:metric-prop}
  \begin{gathered}
    \dR_1 = (\dP'_1\| \dQ'_1), \quad
    \dR_2 = (\dP'_2\| \dQ'_2), \quad
    \dR = (\dP'\|\dQ'), \\
    \TVD(\dP_1',\dP') = \TVD(\dP_1,\dP_3), \quad \TVD(\dQ_1',\dQ') = \TVD(\dQ_1,\dQ_3), \\
    \TVD(\dP_2',\dP') = \TVD(\dP_2,\dP_4), \quad \TVD(\dQ_2',\dQ') = \TVD(\dQ_2,\dQ_4).
    % \MTVD(R,R_1) = \TVD(P_1',P') + \TVD(Q_1',Q'), \quad
    % \MTVD(R,R_2) = \TVD(P_2',P') + \TVD(Q_2',Q').
  \end{gathered}
  \end{equation}
  Then the proof will be concluded by
  \[
  \begin{aligned}
    \MTVD(\dR_1,\dR_2)
    &\leq \max(\TVD(\dP_1',\dP_2'), \TVD(\dQ_1',\dQ_2')) \\
    &\leq \max(\TVD(\dP_1',\dP') + \TVD(\dP_2',\dP'),~ \TVD(\dQ_1',\dQ') + \TVD(\dQ_2',\dQ')) \\
    &\leq \max(\TVD(\dP_1,\dP_3), \TVD(\dQ_1,\dQ_3)) + \max(\TVD(\dP_2,\dP_4), \TVD(\dQ_2,\dQ_4)) \\
    &= \MTVD(\dR,\dR_1) + \MTVD(\dR,\dR_2) .
  \end{aligned}
  \]

  Intuitively, the distribution $\dQ'$ samples a pair of dependent random variables $(X,Y)\sim\dQ'$,
  such that $X\sim \dQ_3$ and the distribution $Y$ conditional on $X=x$
  is the same as the distribution of $Y$ conditional on $\frac{\dP_3(x)}{\dQ_3(x)} = \frac{\dP_4(Y)}{\dQ_4(Y)}$.
  % So does distribution $\dP'$.
  To formalize the intuition, we define a joint distribution $\dQ_{3,4}$.
  \[
    \dQ_{3,4} (x,y) = 
    \begin{cases}
      \dQ_3(x)  \dfrac{\dQ_4(y)}{\dQ_4(\Omega_{2,r})}, &\text{ if } \exists r \text{ s.t.~} x\in \Omega_{1,r} \wedge y\in \Omega_{2,r} \wedge \dQ_4(\Omega_{2,r})>0\\
      0, &\text{ otherwise, }
    \end{cases} 
  \]
  where $\Omega_{1,r},\Omega_{2,r}$ are defined as 
  \[
  \begin{aligned}
    \Omega_{1,r}~ &\deq \Bigl\{ x\in\Omega_1 \Bigm| \frac{\dP_3(x)}{\dQ_3(x)} = r \Bigr\}, \quad
    &\Omega_{2,r}~ &\deq \Bigl\{ y\in\Omega_2 \Bigm| \frac{\dP_4(y)}{\dQ_4(y)} = r \Bigr\},\\
    \Omega_{1,\infty} &\deq \Bigl\{ x\in\Omega_1 \Bigm| {\dQ_3(x)} = 0 \Bigr\},  \quad 
    &\Omega_{2,\infty} &\deq \Bigl\{ y\in\Omega_2 \Bigm| {\dQ_4(y)} = 0 \Bigr\} \,.
  \end{aligned}
  \]
  Note that the probability mass function of $\dQ_{3,4}$ equals 
  \[
    \dQ_3(x)  \dfrac{\dQ_4(y)}{\dQ_4(\Omega_{2,r})}
    = \dR(r) \dfrac{\dQ_3(x)}{\dQ_3(\Omega_{1,r})}  \dfrac{\dQ_4(y)}{\dQ_4(\Omega_{2,r})}
    = \dQ_4(y)  \dfrac{\dQ_3(x)}{\dQ_3(\Omega_{1,r})}.
  \]
  Thus there are two other equivalent interpretations:
  \begin{itemize}
    \item First sample $Y\sim \dQ_4$.  
      Conditioning on $Y=y$, the distribution of $X$ is the same
      as the distribution of $X$ conditional on $\frac{\dP_3(X)}{\dQ_3(X)} = \frac{\dP_4(y)}{\dQ_4(y)}$.
    \item First sample $R\sim \dR$.
      Conditioning on $R=r$,  $X$ and $Y$ are independent, and the conditional distribution of $X$ (resp.~$Y$) is the same as 
      the distribution of $X\sim\dQ_3$ (resp.~$Y\sim\dQ_4$) conditional on $\frac{\dP_3(X)}{\dQ_3(X)} = r$ (resp.~$\frac{\dP_4(Y)}{\dQ_4(Y)}=r$).
  \end{itemize}

  Similarly, we can define a joint distribution $\dP_{3,4}$ as
  \[
    \dP_{3,4} (x,y) = 
    \begin{cases}
      \dP_3(x)  \dfrac{\dP_4(y)}{\dP_4(\Omega_{2,r})}, &\text{ if } \exists r \text{ s.t.~} x\in \Omega_{1,r} \wedge y\in \Omega_{2,r} \wedge \dP_4(\Omega_{2,r})>0\\
      0, &\text{ otherwise. }
    \end{cases} 
  \]
  Note that, if letting $(\dR^{c},\dR)$ be the canonical pair of $\dR$,
  the probability mass function of $\dP_{3,4}$ equals 
  \[
    \dP_3(x)  \dfrac{\dP_4(y)}{\dP_4(\Omega_{2,r})}
    = \dR^c(r) \dfrac{\dP_3(x)}{\dP_3(\Omega_{1,r})}  \dfrac{\dP_4(y)}{\dP_4(\Omega_{2,r})}
    = \dP_4(y)  \dfrac{\dP_3(x)}{\dP_3(\Omega_{1,r})}.
  \]
  % Thus there are three equivalent interpretations of $(X,Y)\sim\dP_{3,4}$:
  % \begin{itemize}
  %   \item First sample $Y\sim \dP_4$.  
  %     Conditioning on $Y=y$, the distribution of $X$ is the same
  %     as the distribution of $X$ conditional on $\frac{\dP_3(X)}{\dQ_3(X)} = \frac{\dP_4(y)}{\dQ_4(y)}$.
  %   \item First sample $R^c\sim \dR^c$.
  %     Conditioning on $R^c=r$,  $X$ and $Y$ are independent, and the conditional distribution of $X$ (resp.~$Y$) is the same as 
  %     the distribution of $X\sim\dP_3$ (resp.~$Y\sim\dP_4$) conditional on $\frac{\dP_3(X)}{\dQ_3(X)} = r$ (resp.~$\frac{\dP_4(Y)}{\dQ_4(Y)}=r$).
  % \end{itemize}

  We can define the conditional distributions $\dP_{3|4},\dP_{4|3},\dQ_{3|4},\dQ_{4|3}$
  so that 
  $\dP_{3,4} = \dP_3 \dP_{4|3} = \dP_4 \dP_{3|4}$
  and $\dQ_{3,4} = \dQ_3 \dQ_{4|3} = \dQ_4 \dQ_{3|4}$.
  That is, for any $(x,y) \in \Omega_1\times \Omega_2$
  \begin{align*}
    \dP_{3,4}(x,y) &= \dP_3(x) \cdot \dP_{4|3} (y|x) = \dP_4 (y) \cdot \dP_{3|4} (x|y), \\
    \dQ_{3,4}(x,y) &= \dQ_3(x) \cdot \dQ_{4|3} (y|x) = \dQ_4 (y) \cdot \dQ_{3|4} (x|y).
  \end{align*}
  Moreover, for any $x\in \supp(\dP_3) \cap \supp(\dQ_3)$, we have $x\in\Omega_{1,r}$ for some $r\in(0,\infty)$,
  then
  \[
    \dP_{4|3}(y|x) 
    = 
    \begin{cases}
      \dfrac{\dP_4(y)}{\dP_4(\Omega_{2,r})}, &\text{ if }y\in\Omega_{2,r}\\
      0, &\text{ otherwise }\\
    \end{cases}
    =
    \begin{cases}
      \dfrac{\dQ_4(y)}{\dQ_4(\Omega_{2,r})}, &\text{ if }y\in\Omega_{2,r}\\
      0, &\text{ otherwise }\\
    \end{cases}
    =
    \dQ_{4|3}(y|x).
  \]
  Thus we can define them so that $\dP_{4|3} = \dQ_{4|3}$.
  Similarly, we can let $\dP_{3|4} = \dQ_{3|4}$.
  
  Define distributions $\dP',\dQ',\dP'_1,\dQ'_1,\dP'_2,\dQ'_2$ over $\Omega = \Omega_1\times\Omega_2$ as
  \[
  \begin{aligned} 
    \dP' &= \dP_{3,4}\,, & \dP'_1 &= \dP_1\dP_{4|3}\,, & \dP'_2 &= \dP_2\dP_{3|4}\,,  \\
    \dQ' &= \dQ_{3,4}\,, & \dQ'_1 &= \dQ_1\dQ_{4|3}\,, & \dQ'_2 &= \dQ_2\dQ_{3|4}\,.  
  \end{aligned}
  \]
  We conclude the proof by verifying the properties in \eqref{eq:metric-prop} holds.

  \begin{itemize}
    \item 
      Let $(X,Y) \sim \dQ'_1$,
      then $\frac{\dP'_1(X,Y)}{\dQ'_1(X,Y)} \sim (\dP'_1\| \dQ'_1)$.
      In the meanwhile,
      \[
        \frac{\dP'_1(X,Y)}{\dQ'_1(X,Y)}
        = \frac{(\dP_1\dQ_{4|3})(X,Y)}{(\dQ_1\dQ_{4|3})(X,Y)}
        = \frac{\dP_1(X)\dQ_{4|3}(Y|X)}{\dQ_1(X)\dQ_{4|3}(Y|X)}
        =   \frac{\dP_1(X)}{\dQ_1(X)}
        \sim \dR_1 \,.
      \]
      Thus $(\dP'_1\| \dQ'_1) = \dR_1$.
    \item By the data processing inequality for TV distance (use it twice)
      \[
        \TVD(\dP_1',\dP')
        = \TVD(\dP_1\dP_{4|3},\dP_{3,4})
        = \TVD(\dP_1\dP_{4|3},\dP_3\dP_{4|3})
        = \TVD(\dP_1,\dP_3)
        \,.
      \]
  \end{itemize}
  The rest of properties in \eqref{eq:metric-prop} can be verified similarly.
\end{proof}

\begin{proof}[Proof of \Cref{lem:error-MTVD}]
  For any $\dP_1,\dQ_1,\dP_2,\dQ_2$ such that $\dR_1 = (\dP_1 \| \dQ_1)$, $\dR_2 = (\dP_2 \| \dQ_2)$,
  we have 
  \[
    |\TVD(\dR_1) - \TVD(\dR_2)|
    = |\TVD(\dP_1,\dQ_1) - \TVD(\dP_2,\dQ_2)|
    \leq \TVD(\dP_1,\dP_2) + \TVD(\dQ_1,\dQ_2)
  \]
  by triangle inequality.
  Thus
  \begin{equation*}
    |\TVD(\dR_1) - \TVD(\dR_2)|
    \leq \inf_{\substack{(\dP_1 \| \dQ_1) = \dR_1 \\ (\dP_2 \| \dQ_2) = \dR_2}} \Bigl( \TVD(\dP_1,\dP_2) + \TVD(\dQ_1,\dQ_2) \Bigr)
    \leq 2\MTVD(\dR_1,\dR_2). \qedhere
  \end{equation*}
\end{proof}

\begin{proof}[Proof of \Cref{prop-cross}]
  Consider any distributions $\dP_1,\dQ_1,\dots,\dP_4,\dQ_4$ satisfying $\dR_i = (\dP_i \| \dQ_i)$ for all $i\in[4]$.
  By~\Cref{prop-product}, 
  $\dR_1 \indpprod \dR_2 = ( \dP_1 \dP_2 \| \dQ_1 \dQ_2 )$ and 
  $\dR_3 \indpprod \dR_4 = ( \dP_3 \dP_4 \| \dQ_3 \dQ_4 )$. 
  Thus 
  \[
    \MTVD(\dR_1 \indpprod \dR_2, \dR_3 \indpprod \dR_4 ) 
    \leq \max( \TVD(\dP_1 \dP_2,  \dP_3 \dP_4 ),~ \TVD(\dQ_1 \dQ_2,  \dQ_3 \dQ_4)).
  \]

  By the the triangle inequality,
  \[
  \TVD(\dP_1 \dP_2,  \dP_3 \dP_4 ) 
  \leq \TVD(\dP_1 \dP_2,  \dP_3 \dP_2 )  + \TVD(\dP_3 \dP_2,  \dP_3 \dP_4 ) 
  =  \TVD(\dP_1,  \dP_3) + \TVD(\dP_2,\dP_4)\,.
  \]
  Similarly, $\TVD(\dQ_1 \dQ_2,  \dQ_3 \dQ_4) \leq \TVD(\dQ_1, \dQ_3) + \TVD(\dQ_2, \dQ_4)$.
  So
  \[
    \MTVD(\dR_1 \indpprod \dR_2, \dR_3 \indpprod \dR_4 ) 
    \leq \max(\TVD(\dP_1, \dP_3), \TVD(\dQ_1, \dQ_3)) + \max(\TVD(\dP_2, \dP_4), \TVD(\dQ_2, \dQ_4)).
  \]

  Since the inequality holds for any $\dP_1,\dQ_1,\dots,\dP_4,\dQ_4$ satisfying $\dR_i = (\dP_i \| \dQ_i)$,
  \[
  \begin{aligned}[b]
    &\MTVD(\dR_1 \indpprod \dR_2, \dR_3 \indpprod \dR_4 ) \\
    &\leq \inf_{\substack{(\dP_1\|\dQ_1) = \dR_1\\(\dP_3\|\dQ_3) = \dR_3}} \max\bigl(\TVD(\dP_1, \dP_3), \TVD(\dQ_1, \dQ_3) \bigr)
    + \inf_{\substack{(\dP_2\|\dQ_2) = \dR_2\\(\dP_4\|\dQ_4) = \dR_4}} \max\bigl(\TVD(\dP_2, \dP_4) + \TVD(\dQ_2, \dQ_4)\bigr) \\
    &= \MTVD(\dR_1, \dR_3) + \MTVD(\dR_2, \dR_4) \,.
  \end{aligned}
  \qedhere
  \]
\end{proof}

\begin{proof}[Proof of \Cref{lem:region-triangle}]
  % The lemma has been implicitly proved when showing the triangle inequality of MTV distance in the proof of \Cref{lem-metric}.
  In the proof of \Cref{lem-metric}, we have shown the following statement.
  \begin{quotation}
  For any distributions $\dP_1,\dQ_1,\dP_3,\dQ_3$ over sample space $\Omega_1$ 
  and distributions $\dP_2,\dQ_2,\dP_4,\dQ_4$ over sample space $\Omega_2$
  satisfying $(\dP_1\|\dQ_1) = \dR_1, (\dP_2\|\dQ_2) = \dR_2, (\dP_3\|\dQ_3) = (\dP_4\|\dQ_4) = \dR$,
  there exist distributions $\dP',\dQ',\dP'_1,\dQ'_1,\dP'_2,\dQ'_2$ 
  satisfying \eqref{eq:metric-prop}.
  \end{quotation}

  \Cref{lem:region-triangle} follows almost directly from the statement,
  by letting 
  \begin{align*}
    \RTVD(\dR_1,\dR) \ni (\TVD(\dP_1,\dP_3),\TVD(\dQ_1,\dQ_3)) \leq  (\alpha_1,\beta_1)\,, \\
    \RTVD(\dR_2,\dR) \ni (\TVD(\dP_2,\dP_4),\TVD(\dQ_2,\dQ_4)) \leq  (\alpha_2,\beta_2)\,,
  \end{align*}
  then the existing of $\dP',\dQ',\dP'_1,\dQ'_1,\dP'_2,\dQ'_2$ satisfying \eqref{eq:metric-prop} implies
  \begin{multline*}
    \RTVD(\dR_1,\dR_2) 
    \ni (\TVD(\dP'_1,\dP'_2),\TVD(\dQ'_1,\dQ'_2)) \\
    \leq  (\TVD(\dP'_1,\dP') + \TVD(\dP',\dP'_2), \TVD(\dQ'_1,\dQ') + \TVD(\dQ',\dQ'_2)) 
    \leq (\alpha_1 + \alpha_2,\beta_1 + \beta_2) \,.
    \qedhere  
  \end{multline*}
\end{proof}
\section{Visualize the Sparsification Process}

This section provides an informal intuitive visualization of the sparsification process.
The Neyman-Pearson region of a decision problem $(\dP,\dQ)$ is considered 
when proving the positivity of the minimum total variation distance (Lemma~\ref{lem-metric}).
In that proof, we claim that the Neyman-Pearson region contains exactly the same information as the ratio $\dR \deq (\dP\|\dQ)$.

Say the ratio $\dR \deq (\dP\|\dQ)$ is represented by a sorted table $(r_1,p_1),(r_2,p_2),(r_3,p_3),\dots$,
that is, $r_1<r_2<\dots$ and $\dR(r_i) = p_i$.
% Say the support of the ratio $\dR \deq (\dP\|\dQ)$ consists of values $r_1,r_2,r_3,\dots$, in that order.
By the Neyman-Pearson lemma,
the (upper) boundary of the Neyman-Pearson region is the polygonal chain that connects
% \[
%   (0,0), \underset{\substack{=(\dQ(\leq r_1),\dP(\leq r_1))}}{(p_1,r_1p_1)}, \underset{\substack{=(\dQ([0,r_2]),\dP([0,r_2]))}}{(p_1+p_2, r_1p_1+r_2p_2)}, \dots, ({\textstyle \sum_{t<i}p_t, \sum_{t<i}r_tp_t}), \dots, (1,1)
% \]
\[
  (0,0), \underset{\substack{=(r_1p_1,p_1)}}{(\dP([0,r_1]),\dQ([0,r_1]))}, \underset{\substack{=(r_1p_1+r_2p_2, p_1+p_2)}}{(\dP([0,r_2]),\dQ([0,r_2]))}, \dots, \underset{=({\sum_{t<i}r_tp_t, \sum_{t<i}p_t})}{(\dP([0,r_t],\dQ([0,r_t])))}, \dots, (1,1) \,,
\]
% \[
%   (0,0), \underset{\substack{=(p_1,r_1p_1)}}{(\dQ([0,r_1]),\dP([0,r_1]))}, \underset{\substack{=(p_1+p_2, r_1p_1+r_2p_2)}}{(\dQ([0,r_2]),\dP([0,r_2]))}, \dots, \underset{=({\sum_{t<i}p_t, \sum_{t<i}r_tp_t})}{(\dQ([0,r_t],\dP([0,r_t])))}, \dots, (1,1)
% \]
as illustrated in \Cref{fig:NP}.
The $i$-th direct segment in the polygonal chain is the vector $(r_ip_i,p_i)$.
This also explains how the Neyman-Pearson region contains exactly the same information as the ratio.

The goal of our sparsification process is to find a simpler ratio (the sparsified ratio) that is close to the given ratio.
From another point of view, the goal is to find a region with simpler boundary that is close to given Neyman-Pearson region. 

If a few consecutive direct segments in the boundary have similar slopes, 
they can be replaced by a single direct segment, causing minor change on the shape of the region.
This is exactly how our sparsification algorithm works:
identifying a cluster of probability masses then merging them together.

\begin{figure}[htp]
\centering
  \includegraphics{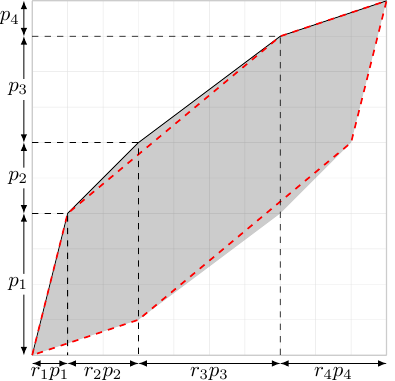}
  \caption{Illustration of the Neyman-Pearson region when ratio has support size 4.
  The red dashed line marks the boundary of the new Neyman-Pearson region of the sparsified ratio.}
  \label{fig:NP}
\end{figure}

% Let $\dR \deq (\dP\|\dQ)$.
% Due to the equivalence between $(\dP,\dQ)$ and $(\dR^\dagger,\dR)$, 
% their Neyman-Pearson are the same.
%
% We have claimed that the Neyman-Pearson region contains exactly the same information as the ratio $\dR \deq (\dP\|\dQ)$.
% This claims follows from the Neyman-Pearson lemma,
% which says: the (lower) boundary of the Neyman-Pearson region is achieved by the likelihood ratio test distinguishers.

\end{document}